\documentclass[journal,draftclsnofoot,onecolumn, 12pt]{IEEEtran} %,draftclsnofoot,onecolumn, 
\ifCLASSINFOpdf
  \usepackage[pdftex]{graphicx}
\else
  \usepackage[dvips]{graphicx}
\fi
\usepackage{float}
\usepackage{amsmath,amsfonts}
\usepackage{algorithm}
\usepackage{diagbox}
\usepackage{algpseudocode}
\usepackage{lipsum}
\usepackage{tabularx} 
\usepackage{arydshln}
\usepackage{tikz}
\usepackage{array}
\makeatletter

   \ExplSyntaxOn
  \prg_new_conditional:Npnn \int_if_zero:n #1 { p , T , F , TF }
  {
  	\if_int_compare:w \int_eval:w #1 = \c_zero_int
  	\prg_return_true:
  	\else:
  	\prg_return_false:
  	\fi:
  }
  \ExplSyntaxOff
  
\makeatother

\usepackage{array}
\usepackage[caption=false]{subfig}
\usepackage{textcomp}
\usepackage{stfloats}
\usepackage{amsthm}
\usepackage{amssymb}
\newtheorem{lemma}{Lemma}
\newtheorem{theorem}{Theorem}

\theoremstyle{definition}
\newtheorem{example}{Example}
\usepackage{url}
\usepackage{verbatim}
\usepackage{mathrsfs}
\usepackage{cite}
\usepackage{bm}
\usepackage{tikz}
\usepackage{pgfplots}
\pgfplotsset{compat=newest}
\usepackage{ifthen}
\usetikzlibrary{arrows.meta, shapes.geometric, decorations.pathmorphing, backgrounds, positioning, fit, petri, calc}
\usepackage{pgffor}
\usetikzlibrary{plotmarks}
\usepgfplotslibrary{patchplots}
\usepackage{grffile}
\usepackage{tablefootnote}
\usepackage{nicematrix}
\usepackage[multiple]{footmisc}

\begin{document}

\title{Coding Theorem for Generalized Reed-Solomon Codes}

\author{Xiangping Zheng and Xiao Ma,~\IEEEmembership{Member,~IEEE}
        % <-this % stops a space
\thanks{Corresponding author: Xiao Ma.}% <-this % stops a space
\thanks{The authors are with the School of Computer Science and Engineering, and also with
the Guangdong Key Laboratory of Information Security Technology, Sun Yat-sen University,  Guangzhou 510006, China (e-mail: zhengxp23@mail2.sysu.edu.cn; maxiao@mail.sysu.edu.cn).}}

\markboth{Journal of \LaTeX\ Class Files,~Vol.~1, No.~2, December~2023}%
{Shell \MakeLowercase{\rmit{et al.}}: A Sample Article Using IEEEtran.cls for IEEE Journals}

% \IEEEpubid{0000--0000~\copyright~2023 IEEE}
% column for its rm to clear the IEEEpubid mark.

\maketitle

\begin{abstract}
In this paper, we prove that the sub-field images of generalized Reed-Solomon~(RS) codes can achieve the symmetric capacity of $p$-ary memoryless channels. Unlike the totally random linear code ensemble, as a class of maximum distance separable~(MDS) codes, the generalized RS code ensemble lacks the pair-wise independence among codewords and has non-identical distributions of nonzero codewords. 
%In the derivation of the coding theorem for the $p$-ary images of generalized RS codes, the random coding techniques and the weight spectrum are exploited, leading to an exponential upper bound on the error probability of the generalized RS code in terms of its spectrum.
To prove the coding theorem for the $p$-ary images of generalized RS codes, we analyze the exponential upper bound on the error probability of the generalized RS code in terms of its spectrum using random coding techniques. In the finite-length region, we present an ML decoding algorithm for the generalized RS codes over the binary erasure channels~(BECs). In particular, the algebraic structure of the generalized RS codes allows us to implement the parallel Lagrange interpolation to derive an ordered systematic matrix. Subsequently, we can reconstruct the ML codeword through a change of basis, accelerating the conventional Gaussian elimination~(GE), as validated in the simulation results. Additionally, we apply this decoding technique to the LC-OSD algorithm over the additive white Gaussian noise~(AWGN) channels with binary phase shift keying~(BPSK) modulation and three-level pulse amplitude modulation~(3PAM). Simulation results show that, in the high-rate region, generalized RS codes defined over fields of characteristic three with 3-PAM perform better than those defined over fields of characteristic two with BPSK. 
\end{abstract}


\section{Introduction}
RS codes are optimal in the sense that they are maximum distance separable~(MDS) codes, which have been widely used in various systems, such as deep space communication, satellite communication, and data storage, due to their robust error correction capabilities for random and burst errors~\cite{Intro_RSapp}. An interesting question arises: are RS codes capacity-achieving? This paper provides a positive answer specifically for the sub-field images of generalized RS codes over symmetric memoryless channels. To the best of our knowledge, such a proof is not available in the literature. It is well-established that, without limiting our focus to generalized RS codes, we can demonstrate the existence of capacity-achieving linear codes over binary-input output-symmetric channels~\cite[Theorem 6.2.1]{Gallager1968} by introducing a totally random linear code ensemble. However, the conventional proof using the error exponent~\cite{Gallager1968} is not applicable to the generalized RS codes due to several significant challenges. Firstly, the codewords in the totally random linear code ensemble are pair-wise independent, while those in the generalized RS code ensemble are not.  Specifically, any two distinct codewords in a generalized RS code of length  $n$ and dimension $k$ have a Hamming distance of at least $n-k+1$. Additionally, while the non-zero codewords in the totally random linear code ensemble are identically distributed, this is not the case for the non-zero codewords in the generalized RS code ensemble. In this paper, we address these challenges by classifying the error patterns into subsets according to their Hamming weights, as developed in~\cite{ma2022new}~\cite{wang2024coding}. This proof technique is powerful, enabling us to derive the random error  exponent in terms of weight spectrum and hence to prove the coding  theorem for generalized RS codes.

In spite of their capacity-achieving property, long generalized RS codes lack efficient decoding algorithms. However, in the short-length region, RS codes can be decoded using the ordered statistic decoding~(OSD) algorithm~\cite{OSD1995Fossorier}~\cite{DorshOSD}. The basic idea of the original OSD is to generate a list of candidate codewords by flipping a small number of bits in the most reliable basis~(MRB). To reduce the complexity of the original OSD algorithm by early stopping unlikely unnecessary test patterns, the OSD with local constraints~(LC-OSD) algorithm~\cite{LC_OSD2022}~\cite{LC_OSDljf2023} extends the MRB  by including extra reliable positions. A disadvantage associated with the OSD algorithms is that they need to perform the Gaussian elimination (GE) with complexity of order $\mathcal{O}(k^3).$ Even worse, the GE is a serial algorithm and no efficient parallel implementation is available for a general matrix, inevitably causing extra decoding complexity and latency. This can be alleviated for RS codes, if the MRB is not mandatory as argued in~\cite{ma2024guessing}, by employing the quasi-OSD, where the GE is implemented in parallel by Lagrange interpolation~\cite{QuasiOSD}.
%Furthermore, the GE process can be skipped by  by precalculating and storing multiple systematic generator matrices~\cite{choi2021fast}, and invoking specific decoding conditions~\cite{yue2022ordered}. Actually, by pre-storing a systematic generator matrix, low-complexity reduced GE is sufficient for performing modified OSD~\cite{fossorier2024modified}. 
%
%Despite these improvements, 
%
%due to the ,...  the , the TEP searchprocess in LC-OSD and other variants remains sequential,lacking parallelism. This limitation can lead to high latency and restricted throughput, especially in high-throughput communication scenarios
%

In this paper, we focus on the $p$-ary images of generalized RS codes, which are capacity-achieving~(demonstrated by theoretical proof) in the asymptotic regime as the code length tends to infinity and are capacity-approaching~(demonstrated by the performance bounds and the numerical results) in the finite-length regime. The main contributions of this paper are listed as follows.
\begin{itemize}
	\item Coding theorem: In the infinite-length region, we prove that the generalized RS code is capacity-achieving over $p$-ary symmetric memoryless channels.	
	\item Decoding algorithm: We present an efficient decoding algorithm for computing the exact maximum-likelihood~(ML) generalized RS codeword over the binary erasure channels~(BECs), in which an ordered systematic matrix is generated in parallel by the Lagrange interpolation and the ML codeword can be reconstructed through a change-of-basis on the systematic matrix. This implementation accelerates the conventional GE operation and thus reduces the decoding latency. Additionally, we derive an approximate upper bound on the ML decoding performance of generalized RS codes over the BECs. The tightness of the bound is confirmed by the simulation results.
	\item Applications: We adapt the Lagrange interpolation plus change-of-basis to the OSD-like algorithms over the additive white Gaussian noise~(AWGN) channels with binary phase shift keying~(BPSK) modulation and three-level pulse amplitude modulation~(3PAM). Simulation results demonstrate that, in the high-rate region, generalized RS codes defined over fields of characteristic three with 3-PAM exhibit a significant performance advantage compared with those defined over fields of characteristic two with BPSK, leading to a promising application in the ultra-reliable low-latency communication~(URLLC), especially when the raw data is ternary.
\end{itemize}

The rest of this paper is organized as follows. In Section~II, we present the system model and the coding theorem for the generalized RS codes. In Section~III, we prove that the generalized RS code is capacity-achieving in the infinite-length region by using the random coding exponent, and the generalized RS codes are analyzed by utilizing the weight spectrum and the random coding bound. In Section IV, we focus on the finite-length generalized RS codes. We present an efficient ML decoding algorithm of the generalized RS codes over BECs, where the ML codeword is reconstructed by parallel Lagrange interpolation and change-of-basis instead of the conventional GE method. Accordingly, the proposed decoding algorithm over BECs is extended to the decoding over BPSK-AWGN channels and 3PAM-AWGN channels. Section V concludes this paper.

% In this paper, we exploit the partial error exponent technique~\cite{ma2022new}~\cite{wang2024coding} to prove the coding theorem for generalized RS codes.

Notation:  A random variable is denoted by an upper-case letter, say $X$, whose realization is denoted by the corresponding lower-case letter $x\in \mathcal{X}$. We use $P_{X}(x)$, $x\in \mathcal{X}$ to represent the probability mass~(or density) function of a random discrete~(or continuous) variable. For a vector of length $n$, we represent it as $\bm x = (x_0, x_1, \cdots, x_{n-1})$. More generally, we interchangeably use $x_i$ and $\bm x[i]$ to represent the $i$-th component of $\bm x$~($0\leq i < n$).

\section{The System Model and the Coding Theorem}
\subsection{Generalized RS Codes}
Let $\mathbb{F}_{p}$ be the finite field of size $p$, where $p$ is a prime or a power of a prime. For a positive integer $m$, let $q=p^m$ and $\mathbb{F}_{q}$ be the finite field defined by a primitive polynomial $f(x) = f_0 + f_1 x + \cdots + f_{m-1}x^{m-1} + x^{m} \in \mathbb{F}_p[x]$. Then we can write $\mathbb{F}_q = \{0\} \cup \{\alpha^i: 0\leq i < q-1\}$, where $\alpha$ is a root of $f(x)$. Any element $\beta \in \mathbb{F}_q$ can be represented by a $p$-ary vector $\boldsymbol{b} = (b_{0}, b_{1}, \dots, b_{m-1})\in \mathbb{F}_p^m$, meaning that $\beta = \sum_{0\leq j < m}b_{j} \alpha^j$. We call $\bm b$ the $p$-ary image of $\beta$. This one-to-one correspondence between $\mathbb{F}_p^m$ and $\mathbb{F}_q$ is denoted as $\beta = \phi(\bm b)$ or $\bm b = \phi^{-1}(\beta)$. It can be verified that $\phi^{-1}(\alpha^i \beta) = \phi^{-1}(\beta)\mathbf{A}^i$, where 
$$\mathbf{A} = \begin{pmatrix}
	0 & 1 & 0 & \cdots   & 0\\
	0 & 0 & 1 & \cdots  & 0\\
	\vdots & \vdots & \vdots &\ddots & \vdots \\
	0 & 0 & 0& \cdots & 1 \\
	-f_0 & -f_1 & -f_2 & \cdots & -f_{m-1}\\
\end{pmatrix}$$
is the companion matrix of $f(x)$.

For any $k < n \leq q-1$, an RS code of  dimension $k$ and length $n$, denoted by $\mathscr{C}_{\text{RS}}[n,k]$, can be defined by a generator matrix $\mathbf{G}_{\text{RS}}$ of size $k \times n$, 
\begin{equation}\label{G}
	\mathbf{G}_{\text{RS}} = \begin{pmatrix}
		1 & 1 & 1 & \cdots & 1  \\
		1 & \alpha & \alpha^2 & \cdots &  \alpha^{n-1}\\
		1 & \alpha^2 & \alpha^4 & \cdots &  \alpha^{2(n-1)}\\
		\vdots & \vdots & \vdots &\ddots & \vdots \\
		1 & \alpha^{k-1} & \alpha^{2(k-1)} & \cdots &  \alpha^{(k-1)(n-1)}\\
	\end{pmatrix}.
\end{equation}

%the evaluation of $u(x)$ at the zero element is allowed, then it is possible to construct any $\mathscr{C}_{\text{RS}}[n,k]$ with $0<k<n \leq q$.

%说 n k生成矩阵对应多项式。。。估值，若允许0。。。则可以构造任意0<k<n<=q...

A generalized RS~(coset) code of  dimension $k$ and length $n$ specified by ($\alpha^{j_0}, \cdots, \alpha^{j_{n-1}}$) and $\bm a \in \mathbb{F}_q^n$ is defined by $\bm a + \mathscr{C}_{\text{GRS}}[n,k] = \{\bm c =  \bm u \mathbf{G}_{\text{GRS}} + \bm a|\bm u \in \mathbb{F}_q^k \}$, where  
%denoted by $\mathscr{C}_{\text{GRS}}[n,k]$, can be defined by a generator matrix $\mathbf{G}_{\text{GRS}}$ of size $k \times n$, 
\begin{equation}\label{GRS_G}
	\mathbf{G}_{\text{GRS}} = \begin{pmatrix}
		\alpha^{j_0} & \alpha^{j_1} & \cdots & \alpha^{j_{n-1}}  \\
		\alpha^{j_0} & \alpha^{j_1} \alpha & \cdots &  \alpha^{j_{n-1}}\alpha^{n-1}\\
		\alpha^{j_0} & \alpha^{j_1}\alpha^2  & \cdots & \alpha^{j_{n-1}}\alpha^{2(n-1)}\\
		\vdots & \vdots  &\ddots & \vdots \\
		\alpha^{j_0} & \alpha^{j_1}\alpha^{k-1} & \cdots &  \alpha^{j_{n-1}}\alpha^{(k-1)(n-1)}\\
	\end{pmatrix},
\end{equation} 
where $0\leq  j_i < q-1$ for $0\leq i < n$. Given a vector $\boldsymbol{u}=(u_0,u_1,\dots, u_{k-1})\in \mathbb{F}_q^k$ and its associated polynomial $u(x) = u_0 + u_1 x + \cdots + u_{k-1} x^{k-1} \in \mathbb{F}_q[x]$, encoding $\boldsymbol{u}$ into $\bm u \mathbf{G}_{\text{GRS}} +\bm a \in \mathbb{F}_q^n$ is equivalent to first evaluating $u(x)$ at $x \in\{ \alpha^j, 0 \leq j \leq n-1\}\subseteq  \mathbb{F}_q$ and then scaling by ($\alpha^{j_0}, \cdots, \alpha^{j_{n-1}}$) and shifting by the vector $\bm a$.
%the componentwise multiplication of $(\alpha^{j_0} ,\alpha^{j_1},\cdots, \alpha^{j_{n-1}})$ with the $n$ evaluations of $u(x)$ at $x \in\{ \alpha^j, 0 \leq j \leq n-1\}\subseteq  \mathbb{F}_q$.
 If we extend the evaluation of $u(x)$ at $x = 0$, we can construct $\mathscr{C}_{\text{GRS}}[n,k]$ for any $0<k<n \leq q$.  
	
	The $p$-ary image of $\mathscr{C}_{\text{GRS}}$, denoted by $\mathscr{C}^{(p)}_{\text{GRS}}$, can be  obtained by mapping  each codeword $\bm c \in \mathscr{C}_{\text{GRS}}$ into $\phi^{-1}(\bm c)$. Here we use by convention $\phi^{-1}(\bm c)$ to denote the component-wise mapping induced by $\phi^{-1}$. Let $\mathbf{G}_{\text{GRS}}^{(p)}$ be the $p$-ary matrix of size $K \times N$, where $K=km$ and $N = nm$, obtained by replacing an element $\alpha^i$ in $\mathbf{G}_{\text{GRS}}$ with the corresponding $p$-ary matrix $\mathbf{A}^i$. It can be checked that $\mathbf{G}_{\text{GRS}}^{(p)}$ is a generator matrix of $\mathscr{C}^{(p)}_{\text{GRS}}$.

\subsection{System Model}
\begin{figure}[!t]
	\centering
	\includegraphics[width=5.8in]{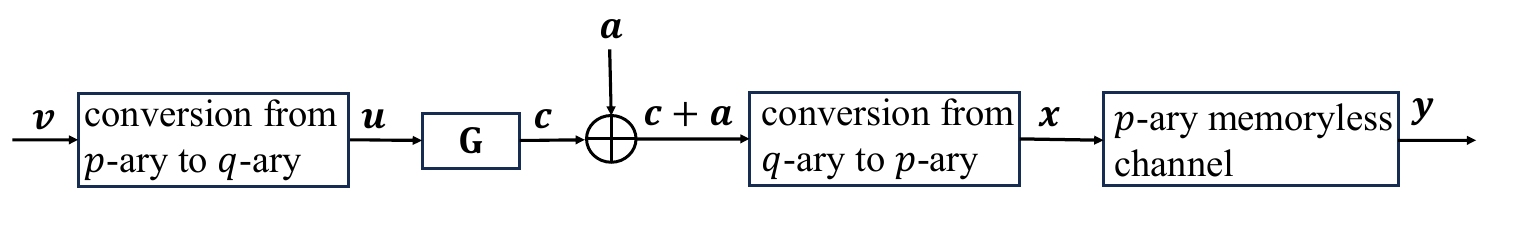}
	\caption{A system model with linear coding, where $\mathbf{G}$ is the $q$-ary generator matrix of size $k\times n$, the conversion from $p$-ary to $q$-ary is a component-wise mapping $\phi(\cdot)$ and the conversion from $q$-ary to $p$-ary is the inverse mapping.}
	\label{system_model}
\end{figure}

 In this paper, we consider a system model that is depicted in Fig.~\ref{system_model}, where $\boldsymbol{v} \in \mathbb{F}_p^K$ is referred to as the message vector, which is interpreted as an equivalent vector $\bm u = \phi(\bm v) \in \mathbb{F}_q^k$ and encoded into a generalized RS codeword $\bm c = \bm u\mathbf{G}_{\text{GRS}} \in \mathbb{F}_q^n$. Then the vector $\bm x = \phi^{-1}(\bm c + \bm a) \in \mathbb{F}_p^N$  is transmitted over a $p$-ary memoryless channel, resulting in  $\boldsymbol{y} \in \mathcal{Y}^N$, where $\mathcal{Y}$ is the alphabet of the channel outputs. The task of the receiver is to recover message vector $\bm v$  from $\bm y$.  The ML decoding is to find a message vector $\hat{\bm v}$ such that 
$P_{Y|X}\left(\bm y|\hat{\bm v}\mathbf{G}_{\text{GRS}}^{(p)} + \phi^{-1}(\bm a)\right)$ $ \geq P_{Y|X}\left(\bm y |\bm v\mathbf{G}_{\text{GRS}}^{(p)}+ \phi^{-1}(\bm a)\right)$ for all $\bm v \in \mathbb{F}_p^K$, where $P_{Y|X}(\cdot|\cdot)$ is the conditional probability mass~(or density) function specifying the channel. The frame error rate~(FER) is defined as $\text{Pr}\{\bm \hat{\bm V} \neq \bm V \}$, where $\hat{\bm V}$ is the estimate of $\bm V$ from decoding. The main coding theorem for generalized RS codes is presented as follows.
\begin{theorem}
	\label{theorem_for_block_code} Let $R$ be a positive number such that $R<I(X;Y)$, where $I(X;Y)$ is the mutual information for the $p$-ary memoryless channel with uniform input distribution. For an arbitrarily small positive number $\epsilon$, one can always find a generalized RS code~(over a sufficiently large field) of length $n$ and dimension $k$ such that its $p$-ary image satisfies that $K/N \geq R$ and that the FER  $\leq \epsilon$.
\end{theorem}

\section{Proof of Main Coding Theorem}
\subsection{Generalized RS Code Ensemble}
To prove the coding theorem for generalized RS codes over $p$-ary memoryless channels, we introduce the following generalized RS code ensemble.

\emph{Generalized RS code ensemble:} A generalized RS code ensemble transforms $\bm u$ into $\bm u \mathbf{G}+\bm a \in \mathbb{F}_q^n$ with the generator matrix $\mathbf{G} = \mathbf{G}_{\text{RS}} \mathbf{\Lambda} \mathbf{\Pi}$ of size $k\times n$, where $\mathbf{\Lambda}$ is a diagonal matrix of order $n$ with elements on diagonal being drawn independently and uniformly from $\mathbb{F}_q \setminus \{0\}$, $\mathbf{\Pi} $ is a permutation matrix of order $n$ being generated uniformly at random, and $\bm a  \in \mathbb{F}_q^n$ is a random vector with each component being generated independently and uniformly from $\mathbb{F}_q$.

It is well known that Theorem~\ref{theorem_for_block_code} holds if we do not restrict our attention to the generalized RS codes. Actually, the existence of a $p$-ary linear code as required by Theorem~\ref{theorem_for_block_code} can be proved in a similar way to~that for~\cite[Theorem 6.2.1]{Gallager1968} by introducing a totally random linear code ensemble. However, the conventional proof is not applicable to the generalized RS code ensemble. This is because the codewords in the totally random linear code ensemble are pair-wise independent, while the codewords in the generalized RS code ensemble are not. Actually, any two distinct codewords in the generalized RS code ensemble have a Hamming distance of at least $n-k+1$. Given that the message vector $\bm 0 \in \mathbb{F}_q^k$ is encoded into a random vector $\bm a$, the non-zero message vector $\bm u \in \mathbb{F}_q^k$ with $W_H(\bm u \mathbf{G}_{\text{RS}}) = w$~($n-k+1\leq w \leq n$) is encoded into a random vector $\bm u\mathbf{G}+\bm a$ distributed as, for any vector  $\bm c  \in \mathbb{F}_q^n$ with $W_H(\bm c) = w$,
\begin{equation}\label{Wdifferent}
	P(\bm u\mathbf{G} + \bm a= \bm c + \bm a| \bm a) = 	P(\bm u\mathbf{G}= \bm c) = \frac{1}{(q-1)^w \binom{n}{w}}.
\end{equation}

%\begin{itemize}
%	\item The codewords in the totally random linear code ensemble are pair-wise independent, while the codewords in the generalized RS code ensemble are not. Actually, any two distinct codewords in the generalized RS code ensemble have a Hamming distance of at least $n-k+1$. 
%	%Particularly, a nonzero codeword in the generalized RS code ensemble must have a Hamming weight greater than or equal to $n-k+1$.
%%	\item The non-zero codewords in the totally random linear code ensemble are identically distributed, while the non-zero codewords in the generalized RS code ensemble are not. 
%%	Actually, if a vector $\bm u$ is encoded into an RS codeword of Hamming weight $W_H(\bm u \mathbf{G}_{\text{RS}}) = w$, then it will be encoded into a random vector $\bm u \mathbf{G}$ in the generalized RS code ensemble such that, for any vector $\bm c$ with $W_H(\bm c) = w$,
%%	\begin{equation}\label{Wdifferent}
%%		P(\bm u\mathbf{G}= \bm c) = \frac{1}{(q-1)^w \binom{n}{w}}.
%%	\end{equation}
%\end{itemize}

\subsection{Random Coding Exponent}

\begin{lemma}
	\label{error_exponent} 
	Let a $p$-ary discrete memoryless channel have transition probabilities $P(j|i)$. For any positive integer $N$ and positive number $R$, an $(N,R)$ block code is a code of block-length $N$ with $\lceil p^{NR}\rceil$ codewords. Consider the ensemble of $(N,R)$ block codes in which each letter of each codeword is independently selected with the probability assignment $Q(i)$. Then, for $0\leq \rho \leq 1$, the ensemble average probability of decoding error using ML decoding satisfies    
	\begin{equation}
	\begin{aligned}
		&\emph{{\text{Pr}}}\{{\rm error}\} \leq p^{NR\rho} \left\{\sum\limits_{j} \left\{\sum\limits_{i} Q(i) (P(j|i))^{1/(1+\rho)}\right\}^{1+\rho} \right\}^{N}.
	\end{aligned}
\end{equation}
Therefore
	\begin{equation}
		\begin{aligned}
			\label{conditional_pro}
			&\emph{{\text{Pr}}}\{{\rm error}\} \leq \exp\Big[-NE_r(R)\Big]\text{,}
		\end{aligned}
	\end{equation}
where $\text{exp}$ is denoted as the base-$p$ exponent, 
	\begin{equation}
		E_r(R)=\max\limits_{0\leq\rho\leq 1}(E_0(\rho,\bm Q)-\rho R)\text{,}
	\end{equation}
	and
	\begin{equation}
		\begin{aligned}
			E_0(\rho,\bm Q)=&-\log_p\Bigg\{\sum\limits_{j} \left\{\sum\limits_{i} Q(i) (P(j|i))^{1/(1+\rho)}\right\}^{1+\rho}\Bigg\}\text{.}
		\end{aligned}	
		\label{E_independent}
	\end{equation}
	The function $E_r(R)$ is referred to as the \textit{random coding exponent}, which is a decreasing function of $R$ and $E_r(R)>0$ if $R<I(\bm Q;\bm P)$, where $I(\bm Q;\bm P)$ is the mutual information specified by the
	input distribution $\bm Q$ and the channel conditional distribution $\bm P$, as shown in~\cite[pp. 140-144]{Gallager1968}.
%	\begin{itemize}
%		\item With a given $\gamma\in [0,1]$, $E_0(p,\gamma)$ is an increasing function of $p\in [0,1]$.
%		\item With a given  $R$, $E(p,R)$ is an increasing function of  $p\in [0,1]$. 
%		\item With a given $\rho\in [0,1]$,  $E_r(R)$ is a decreasing function of $R$ and $E(p,R)>0$ if $R<I_0(p)$.
%		
%	\end{itemize}	
\end{lemma}

%It is worth noting that Lemma~\ref{partial_error_exponent} holds even if the members in the list are dependent. In particular, the pair-wise independence is not required in the proof of Lemma~\ref{partial_error_exponent}.

% A simple symmetry argument is introduced due to the added random vector $\bm a$. 
\subsection{Proof of Theorem~\ref{theorem_for_block_code}}
Let $\epsilon>0$ be an arbitrarily small number. For any positive integer $n$, we can define a  generalized RS code $\mathscr{C}_{\text{GRS}}[n,k]$ over $\mathbb{F}_{p^m}$ by setting $m=\lceil \log_p (n+1)\rceil$ and $k = \lceil nR \rceil$. It is easy to see that the code rate $K/N$ approaches $R$ from above as $n$ goes to infinity. Due to the randomness of the added vector $\bm a$, we can analyze the performance by assuming, without loss of generality, that the message vector $\bm 0 \in \mathbb{F}_q^k$ is encoded into $\bm a$ and thus $\bm x = \phi^{-1} (\bm a)$ is transmitted. Here, we make the following partition
%Let $\mathcal{A}$ be the set of codewords of the generalized RS code ensemble. That is, $\mathcal{A}=\{\bm u \in \mathbb{F}_q^{k} \} $, which can be written as
\begin{equation}
	\begin{aligned}	  \mathbb{F}_q^{k} \setminus \{\bm 0\} = \bigcup \limits_{n-k+1 \leq w \leq n} \mathcal{A}_w,
	\end{aligned}
\end{equation}
where $	\bm u \in \mathcal{A}_w$ if and only if $W_H(\bm u \mathbf{G}_{\text{RS}}) =w$. It can be checked that $|\mathcal{A}_{w}| \leq q^{k-(n-w)}\binom{n}{w}$ since the message vector can be reconstructed from any $k$ positions of a codeword~($n-w$ zeros and $k-n+w$ nonzeros) for MDS codes. Given a received sequence $\bm y$, denote by $E_{\bm u, \bm a}$ the event that $\bm u\mathbf{G}+\bm a$ is more likely than $\bm a$. For the decoding error, we have 
\begin{equation}
	\label{error0}
	\begin{aligned}
		{\rm{Pr}}\{{\rm error}|\bm x\}
		&=\sum\limits_{\bm a\in \mathbb{F}_q^{n}} \frac{1}{q^n} \sum\limits_{\bm y\in \mathcal{Y}^{N}}P(\bm y|\bm x)\cdot  {\rm {Pr}}\Bigg\{{\bigcup\limits_{\bm u \in \mathbb{F}_q^{k} \backslash \{\bm 0\}}E_{\bm u,\bm a}\Bigg\}}\\
		&\leq \sum\limits_{w=n-k+1}^{n}\sum\limits_{\bm x\in \mathbb{F}_p^{N}} \frac{1}{p^N} \sum\limits_{\bm y\in \mathcal{Y}^{N}}P(\bm y|\bm x)\cdot  {\rm {Pr}}\Bigg\{{\bigcup\limits_{\bm u \in \mathcal{A}_w}E_{\bm u, \bm a}\Bigg\}}
		\text{.}
	\end{aligned}	
\end{equation}
%where $E_i$ is the event that, given a received sequence $\bm y$, $\bm u_i$ is more likely than $\bm 0$.

For each $w$~($n-k+1 \leq w \leq n$), any non-negative number $\rho \leq 1$, and $s = 1/(1+\rho)$, we have,
\begin{align}\label{latterbound1}
	&\sum\limits_{\bm x\in \mathbb{F}_p^{N}} \frac{1}{p^N} \sum\limits_{\bm y\in \mathcal{Y}^{N}}P(\bm y|\bm x)\cdot{\rm {Pr}}\Bigg\{{\bigcup\limits_{\bm u \in \mathcal{A}_{w}}E_{\bm u, \bm a}\Bigg\}} \nonumber \\
	&\leq\sum\limits_{\bm x\in \mathbb{F}_p^{N}} \frac{1}{p^N} \sum\limits_{\bm y\in \mathcal{Y}^{N}}\!P(\bm y|\bm x){\rm Pr}\left\{\bigcup\limits_{ \bm u\in\mathcal{A}_{w}} P(\bm y|\phi^{-1}(\bm u\mathbf{G}+\bm a)=\tilde{\bm x})\geq P(\bm y|\bm x)\right\} \nonumber\\
	&\leq \sum\limits_{\bm x\in \mathbb{F}_p^{N}} \frac{1}{p^N} \sum\limits_{\bm y\in \mathcal{Y}^{N}}\!P(\bm y|\bm x)\left(\sum\limits_{\bm u \in \mathcal{A}_{w}}\!{\rm Pr}\{P(\bm y|\phi^{-1}\!(\bm u\mathbf{G}+\bm a)\!=\! \tilde{\bm x})\geq P(\bm y|\bm x)\}\right)^{\rho} \nonumber\\	
	&\overset{(*)}{\leq} \sum\limits_{\bm x\in \mathbb{F}_p^{N}} \frac{1}{p^N} \sum\limits_{\bm y\in \mathcal{Y}^{N}}\!P(\bm y|\bm x)\left\{|\mathcal{A}_{w}| \frac{1}{(q-1)^w \binom{n}{w}}\!\sum\limits_{\tilde{\bm x}=\phi^{-1}\!(\bm u\mathbf{G}+\bm a)\in \mathbb{F}_p^{N}: \atop W_H(\bm u\mathbf{G}) = w}\!\! \frac{(P(\bm y|\tilde{\bm x}))^{s}}{(P(\bm y|\bm x))^{s}}\right\}^{\rho} \nonumber\\
	&\overset{(**)}{\leq} \widetilde{A}_w^\rho \sum\limits_{\bm x\in \mathbb{F}_p^{N}} \frac{1}{p^N} \sum\limits_{\bm y\in \mathcal{Y}^{N}}\!P(\bm y|\bm x)\left\{\sum\limits_{\tilde{\bm x}\in \mathbb{F}_p^{N}}  \frac{1}{p^N} \frac{(P(\bm y|\tilde{\bm x}))^{s}}{(P(\bm y|\bm x))^{s}}\right\}^{\rho}\nonumber\\
	&= \widetilde{A}_w^\rho  \sum\limits_{\bm y\in \mathcal{Y}^{N}} \sum\limits_{\bm x\in \mathbb{F}_p^{N}} \frac{1}{p^N}\!P(\bm y|\bm x)^{1-s\rho}\left\{\sum\limits_{\tilde{\bm x}\in \mathbb{F}_p^{N}} \frac{1}{p^N} (P(\bm y|\tilde{\bm x}))^{s}\right\}^{\rho}\nonumber\\
	&\overset{(***)}{=} \widetilde{A}_w^\rho  \sum\limits_{\bm y\in \mathcal{Y}^{N}} \left\{\sum\limits_{\bm x\in \mathbb{F}_p^{N}} \frac{1}{p^N} (P(\bm y|\bm x))^{1/(1+\rho)}\right\}^{1+\rho}\nonumber\\
	&= \widetilde{A}_w^\rho\left\{\sum\limits_{y_i\in \mathcal{Y}} \left\{\sum\limits_{x_i\in \mathbb{F}_p} \frac{1}{p} (P(y_i|x_i))^{1/(1+\rho)}\right\}^{1+\rho} \right\}^{N}.
\end{align}
where the inequality $(*)$ follows from~\eqref{Wdifferent} and the Markov inequality, the inequality $(**)$ follows by denoting $\widetilde{A}_w = |\mathcal{A}_{w}| q^n/ ((q-1)^w \binom{n}{w})$ and including more terms in the inner summation, and the equality $(***)$ holds since $\bm x$ is a dummy variable of summation.

Thus, we have 
\begin{equation}
	\begin{aligned}
	\sum\limits_{\bm x\in \mathbb{F}_p^{N}} \frac{1}{p^N} \sum\limits_{\bm y\in \mathcal{Y}^{N}}P(\bm y|\bm x)\cdot{\rm {Pr}}\Bigg\{{\bigcup\limits_{\bm u \in \mathcal{A}_{w}}E_{\bm u, \bm a}\Bigg\}} \leq \exp\Big[-NE_r(\widetilde{R}_w )\Big]\text{,}
	\end{aligned}
\end{equation}
where
\begin{equation}
	\widetilde{R}_w = \frac{1}{N}\log_p\widetilde{A}_w,
\end{equation}
\begin{equation}
	E_r(\widetilde{R}_w )=\max\limits_{0\leq\rho\leq 1}(E_0(\rho)-\rho \widetilde{R}_w )\text{,}
\end{equation}
and
\begin{equation}
	\begin{aligned}
		E_0(\rho)=&-\log_p\Bigg\{\sum\limits_{y_i\in \mathcal{Y}} \left\{\sum\limits_{x_i\in \mathbb{F}_p} \frac{1}{p} (P(y_i|x_i))^{1/(1+\rho)}\right\}^{1+\rho} \Bigg\}\text{.}
	\end{aligned}	
\end{equation}

Defining $\tilde{R} = \frac{1}{m }\log_p\left(1+\frac{1}{p^m-1}\right)+\frac{k}{n}$ and recalling that  $|\mathcal{A}_{w}| \leq q^{k-(n-w)}\binom{n}{w}$, we have 
\begin{align}	
	\widetilde{R}_w
	&\leq \frac{1}{n m  }\log_p\left(\left(1+\frac{1}{q-1}\right)^w q^k\right)\nonumber\\
	&= \frac{w}{n m }\log_p\left(1+\frac{1}{p^m-1}\right)+\frac{k}{n}\nonumber \\
	& \leq \tilde{R}.
\end{align}	

Since the error exponent is a decreasing function of $R$ as stated in Lemma~\ref{error_exponent}, we have 
\begin{equation}\label{error_p}
	\begin{aligned}
		{\rm{Pr}}\{{\rm error}|\bm x\}
		& \leq k\exp\left[{-N E_r\left({\widetilde{R}}\right)}\right]\\
		& = \exp\left[-N\left( E_r\left(\widetilde{R}\right)-\frac{\log_p k}{N}\right)\right].
	\end{aligned}	
\end{equation}

Letting $n \rightarrow \infty$, we have $\widetilde{R} \rightarrow  R < I(X;Y)$ and $ (\log_p k)/N \rightarrow  0$. Hence we have $E_r\left(\widetilde{R}\right)>0$ and ${\rm Pr}\{{\rm error}|\bm x\} \leq \epsilon$ for sufficiently large $n$.

Since the average FER of the generalized RS code ensemble satisfies that $\rm{FER} \leq \epsilon$, there exists a matrix $\mathbf{G} = \mathbf{G}_{\text{RS}} \mathbf{\Lambda} \mathbf{\Pi}$ of size $\lceil nR  \rceil \times n$ and a vector $\bm a$ such that $K/N \geq R$ and that the FER $\leq \epsilon$. Since the channel is memoryless and stationary, we see that $\mathbf{\Pi}$ is not essentially required and can be removed without affecting the FER, indicating that the corresponding generalized RS code $\{\bm c = \bm u \mathbf{G}_{\text{RS}} \mathbf{\Lambda} + \bm a \mathbf{\Pi}^{-1}|\bm u \in \mathbb{F}_q^k\}$ also satisfies that $\rm{FER} \leq \epsilon$. This completes the proof of Theorem 1.

Considering the generalized RS codes defined over finite fields of characteristic two, we have the following theorem.

\begin{theorem}
Let $R$ be a positive number such that $R<C$ where $C$ is the binary-input output-symmetric~(BIOS) channel capacity. For an arbitrarily small positive number $\epsilon$, one can always find a generalized RS code~(over a sufficiently large field) of length $n$ and dimension $k$ such that its binary image satisfies that $K/N \geq R$ and that the FER  $\leq \epsilon$.
%of length $n$ and dimension $k$ satisfying $m =  \lceil \log_2 n \rceil$ and $km = \lceil nmR \rceil$ such that the FER $\leq \epsilon$.
\end{theorem}
\begin{proof}
This is a special case of Theorem 1, and the proof is omitted here. It is worth pointing out that the shifting vector $\bm a$ is not necessary  for BIOS channels.
\end{proof}

\subsection{Random Coding Union Bound}
From the above proof, we have an upper bound for an arbitrary $q$-ary code with a generator matrix $\mathbf{G}$ of size $k \times n$. Let $A_w$, $w = 0,1,\cdots, n$ be its weight spectrum, i.e., $A_w = |\{\bm u \mathbf{G} : W_H(\bm u \mathbf{G} ) = w, \bm u \in \mathbb{F}_q^k\}|$. We can create an ensemble that transforms $\bm u$ into $\bm u \mathbf{G}\mathbf{\Lambda}\mathbf{\Pi} + \bm a$, similar to the generalized RS ensemble. Then $\rm{Pr}\{\rm{error}\}$, the average FER over $p$-ary memoryless channels, is upper-bounded by 
\begin{equation}
	\begin{aligned}
		{\rm{Pr}}\{{\rm error}\}
		&\leq \sum\limits_{w=0}^{n} \exp\left[{-N E_r\left({\widetilde{R}_w}\right)}\right]\\
		& \leq n\exp\left[{-N E_r\left({\widetilde{R}}\right)}\right]\\
		& = \exp\left[-N\left( E_r\left(\widetilde{R}\right)-\frac{\log_p n}{N}\right)\right],
	\end{aligned}	
\end{equation}
where $\widetilde{A}_w = A_{w}q^n / ((q-1)^w \binom{n}{w})$, $	\widetilde{R}_w = \frac{1}{n}\log_p\widetilde{A}_w$, $	\widetilde{R} = \max\limits_{0\leq w \leq n} \widetilde{R}_w$, and
\begin{equation}
	E_r(\widetilde{R}_w )=\max\limits_{0\leq\rho\leq 1}(E_0(\rho)-\rho \widetilde{R}_w )\text{,}
\end{equation}
and
\begin{equation}
	\begin{aligned}
		E_0(\rho)=&-\log \Bigg\{\sum\limits_{y_i\in \mathcal{Y}} \left\{\sum\limits_{x_i\in \mathbb{F}_p} \frac{1}{p} (P(y_i|x_i))^{1/(1+\rho)}\right\}^{1+\rho} \Bigg\}\text{.}
	\end{aligned}	
\end{equation}
The above upper bound, which is similar to that presented for binary codes in~\cite{Shulman1999}, is powerful for proving coding theorem but typically loose in the finite length regime. In this case, we turn to the  random coding union~(RCU) bound.

%As shown in the proof of the Theorem 1, if $A_{w} \leq q^{k-(n-w)}\binom{n}{w}$ holds for all $0\leq w \leq n$,  there exists a matrix $\mathbf{G} \mathbf{\Lambda} \mathbf{\Pi}$ of size $\lceil nR  \rceil \times n$ and a vector $\bm a$ such that $k/n \geq R$ and that the corresponding
%$q$-ary~(coset) code $\{\bm c = \bm u \mathbf{G} \mathbf{\Lambda} + \bm a \mathbf{\Pi}^{-1}|\bm u \in \mathbb{F}_q^k\}$  satisfies $\rm{FER} \leq \epsilon$.

Similar to but different from the RCU bound~\cite{polyanskiy2010channel}, which applies to totally random codes, we derive a RCU bound for the generalized RS code ensemble.

Without loss of generality, suppose that the message vector $\bm 0 \in \mathbb{F}_q^k$ is encoded into $\bm a$ and thus $\bm x = \phi^{-1} (\bm a)$ is transmitted over a $p$-ary memoryless channel, resulting in $\mathbf{Y}$. Given the partition of the message, denote by $  \mathbb{F}_q^{k} \setminus \{\bm 0\} = \bigcup \limits_{n-k+1 \leq w \leq n} \mathcal{A}_w$, where $|\mathcal{A}_w| = A_w$ and $\bm u \in \mathcal{A}_w$ if and only if $W_H(\bm u \mathbf{G}_{\text{RS}}) =w$. The average error probability using ML decoding can be upper bounded by
\begin{equation}
\begin{aligned}
	{\rm{FER}}_{\text{avg}}
	&=\mathbb{E}\left[\min\left\{1, \sum\limits_{\bm u \in \mathbb{F}_q^k\setminus \{\bm 0\}}{\rm{Pr}}\left\{P(\mathbf{Y}|\phi^{-1}(\bm u\mathbf{G} + \boldsymbol{a}))\geq P(\mathbf{Y}|\phi^{-1}(\boldsymbol{a}))
	\right\}
	\right\}\right]\\
	&\leq \mathbb{E}\left[\min\left\{1, \sum\limits_{w=n-k+1}^{n}
	\sum\limits_{\bm u \in \mathcal{A}_w}{\rm{Pr}}\left\{P(\mathbf{Y}|\phi^{-1}(\bm u\mathbf{G} + \boldsymbol{a}))\geq P(\mathbf{Y}|\phi^{-1}(\boldsymbol{a}))
	\right\}
	\right\}\right]\\
		& \leq  \mathbb{E}\left[\min\left\{1, \sum\limits_{w=n-k+1}^{n}
	A_w\sum\limits_{\bm x=\phi^{-1}\!(\bm u\mathbf{G}+\bm a)\in \mathbb{F}_p^{N}: \;W_H(\bm u\mathbf{G}) = w, \atop P(\mathbf{Y}|\bm x)\geq P(\mathbf{Y}|\phi^{-1}(\boldsymbol{a}))}
	\frac{1}{(q-1)^w \binom{n}{w}}	\right\}\right]\\
%	& =  \mathbb{E}\left[\min\left\{1, \sum\limits_{w=n-k+1}^{n}
%\frac{A_w q^w}{(q-1)^w}\sum\limits_{\bm x=\phi^{-1}\!(\bm u\mathbf{G}+\bm a)\in \mathbb{F}_p^{N}: \;W_H(\bm u\mathbf{G}) = w, \atop P(\mathbf{Y}|\bm x)\geq P(\mathbf{Y}|\phi^{-1}(\boldsymbol{a}))}
%\frac{1}{q^w \binom{n}{w}}	\right\}\right]\\
&\leq \mathbb{E}\left[\min\left\{1, \sum\limits_{w=n-k+1}^{n}
A_w {\rm{PEP}}_w(\bm a, \bm Y)
\right\}\right],
\end{aligned}	
\end{equation}
where  ${\rm{PEP}}_w(\bm a, \bm Y)$ is the pairwise error probability, defined as
\begin{equation}
	\begin{aligned}
	{\rm{PEP}}_w(\bm a, \bm Y) \triangleq {\rm{Pr}}\left\{P(\mathbf{Y}|\phi^{-1}(\bm C_w+ \boldsymbol{a}))\geq P(\mathbf{Y}|\phi^{-1}(\boldsymbol{a}))\right\},
	\end{aligned}	
\end{equation}
where $\bm C_w \in \mathbb{F}_q^n$ is a random vector uniformly distributed over the set of all vectors with $W_H(\bm C_w)= w$.

\section{Decoding of Generalized RS Codes}
We have proved that the generalized RS code is (symmetric) capacity-achieving over $p$-ary discrete memoryless channels in the infinite-length region. However, it should be well understood that no efficient decoding algorithm is available for general RS codes. Nonetheless, this does not preclude the existence of efficient decoding algorithms for short generalized RS codes, as well as for long high-rate or low-rate generalized RS codes.

%This motivates us to explore decoding algorithms tailored to these three scenarios.

\subsection{Binary Erasure Channels}
Consider the BECs, where a received bit is either completely known or completely unknown~(erased). We denote the bit erasure probability by $\varepsilon$. The channel capacity is given by $C(\varepsilon) = 1 - \varepsilon$~\cite{thomas2006elements}. Recalling the system model shown in Fig.~\ref{system_model}, we can assume that $\bm a = \bm 0$. Specifically, associated with a vector $\bm u = \phi(\bm v) \in \mathbb{F}_q^k$ with  $q=2^m$ is a generalized RS codeword $\bm c = \bm u\mathbf{G}_{\text{GRS}} \in \mathbb{F}_q^n$. Then the vector $\bm x = \phi^{-1}(\bm c) \in \mathbb{F}_2^N$  is transmitted over a BEC, resulting in  $\boldsymbol{y}$. For the sake of simplicity of notation, we omit the subscript ``GRS" hereafter, using $\mathbf{G}$ to represent $\mathbf{G}_{\text{GRS}}$ and $\mathbf{G}^{(b)}$ to represent $\mathbf{G}^{(b)}_{\text{GRS}}$, where the $\mathbf{G}^{(b)}_{\text{GRS}}$ is the binary matrix of $\mathbf{G}_{\text{GRS}}$. For a vector $\bm c \in \mathbb{F}_q^n$, we interchangeably use $\bm c^{(b)}$ and $\phi^{-1}(\bm c)$ to represent its binary image.

Denote by $\mathcal{L} \subseteq\{0,1,\cdots, N-1\}$ the index set of known bits in $\bm c^{(b)}$~(available from $\bm y$), and by $\mathcal{L}_c$ the index set of erasures. Denote by $\mathbf{G}^{(b)}_{\mathcal{L}}$ and $\mathbf{G}^{(b)}_{\mathcal{L}_c}$ the matrix $\mathbf{G}^{(b)}$ restricted to columns indexed by $\mathcal{L}$ and $\mathcal{L}_c$, respectively. Let $\bm c_{\mathcal{L}}^{(b)}$ be the sub-vector of $\bm c^{(b)}$ consisting of $\bm c^{(b)}[i]$, $i\in \mathcal{L}$. We have the following set of linear equations,
\begin{equation}\label{Hck}
	\begin{aligned}
	\bm v	\mathbf{G}^{(b)}_{\mathcal{L}} = {\bm c_{\mathcal{L}}^{(b)}}. 
	\end{aligned}	
\end{equation}
As a result, it can be seen that ML decoding over the BECs is equivalent to solving~\eqref{Hck} for $\bm v \in \mathbb{F}_2^K$. This can be done by the GE. The main issue of the decoding is the delay caused by the GE, which has no efficient parallel implementation available for general matrices. We now present a much more efficient ML decoding algorithm. 

Upon receiving $\boldsymbol{y}$, we calculate the bit sequence $\hat{\bm x}$ by 
\begin{equation}\label{harddecison}
	\hat{x}_i=
	\begin{cases} 
		y_i, & \mbox{if } i\in \mathcal{L}\\
		e, & \mbox{if }i\in \mathcal{L}_c\\
	\end{cases}
	,~0\leq i < N,
\end{equation}
where $e$ denotes the erased bit. The bit sequence $\hat{\bm x}$ is also written as a symbol sequence $\hat{\bm c}=(\hat{c}_0,\hat{c}_1, \dots, \hat{c}_{n-1})$. If there are more than $k$ symbols completely known in $ \hat{\bm c}$, the message polynomial can be simply reconstructed from these symbols using interpolation. Otherwise, we first employ the parallel Lagrange interpolation to derive an ordered systematic generator matrix and then change the basis by swapping the columns indexed by $\mathcal{L}_c$ in the identity matrix with the columns indexed by $\mathcal{L}$ but not in the identity matrix.

%apply a method similar to  the modified ordered statistic decoding~(OSD) algorithm~\cite{fossorier2024modified} and the iterative basis update~(IBU) method~\cite{li2024iterative}.

For two components $\hat{c}_i$ and $\hat{c}_j$ of the symbol sequence $\hat{\bm c}$, we say that $\hat{c}_i$ is more reliable than $\hat{c}_j$ if $\hat{c}_i$ contains more bits in $\mathcal{L}$ than $\hat{c}_j$ does. In the case when two symbols contain the same number of bits in $\mathcal{L}$, the symbol that comes first in lexicographic order is considered more reliable. Then we sort 
$(0,1,\dots,n-1)$ into $\boldsymbol{p} = (0,1,\dots,n-1)\mathbf{\Pi}_{\bm p} = (p_0,p_1,p_2,\dots, p_{n-1})$ such that the first $k$ positions of $\hat{\bm c}\mathbf{\Pi}_{\bm p}$ correspond to the most reliable symbols. Now our objective is to transform the generator matrix $\mathbf{G}\mathbf{\Pi}_{\bm p}$ into a systematic form, denoted by
\begin{equation}%\hspace{-0.34cm}
	\begin{array}{c p{0.5cm}}
		\text{\footnotesize{column labeling:}}   \\ %\vspace{3.5pt}
		\\ %\vspace{1pt}
		\!\!\!\hspace{0.55cm}\widetilde{\mathbf{G}}=
		\\ 
		\\
		\\
		\\
	\end{array}\hspace{-0.38cm}
	\begin{pNiceMatrix}[first-row]
		p_{0} & p_{1} & \cdots & p_{k-1} & p_{k} & \cdots & p_{n-1} \\
		1 & 0 & \cdots& 0 & g_{0,k}& \cdots & g_{0,n-1}  \\
		0 & 1  & \cdots & 0 & g_{1,k}& \cdots & g_{1,n-1} \\
		\vdots & \vdots &\ddots & \vdots & \vdots &\ddots & \vdots\\
		0 & 0 & \cdots &  1 & g_{k-1,k}& \cdots & g_{k-1,n-1} 
	\end{pNiceMatrix}.
\end{equation}
This transformation is equivalent to finding and evaluating  $k$ message polynomials $u^{(i)}(x) = u^{(i)}_{0} + u^{(i)}_{1} x + \cdots + u^{(i)}_{k-1} x^{k-1}$~($0\leq i<k$), each for one row such that
\begin{equation}
	\begin{cases} 
		\alpha^{j_{p_i}}u^{(i)}(\alpha^{p_i}) = 1,\\
		\alpha^{j_{p_j}}u^{(i)}(\alpha^{p_j}) = 0, \quad j \in \{0,1,\dots,k-1\}\setminus\{i\}\\
	\end{cases}
\end{equation}
and
\begin{equation}\label{lagea}
	g_{i,j} = \alpha^{j_{p_{j}}} u^{(i)}(\alpha^{p_{j}})
\end{equation}
for $k\leq j < n$.

Different from GE, these message polynomials can be calculated in parallel by Lagrange interpolation  as 
\begin{equation}\label{Larpo}
	u^{(i)}(x) = \frac{1}{\alpha^{j_{p_i}}}\frac{\prod \limits_{j=0,j\neq i}^{k-1} (x-\alpha^{p_j})}{\prod \limits_{j=0,j\neq i}^{k-1} (\alpha^{p_i}-\alpha^{p_j})}, \quad 0\leq i < k. 
\end{equation}

Correspondingly, the systematic generator matrix $\widetilde{\mathbf{G}}$ can be written as a binary systematic matrix $\widetilde{\mathbf{G}}^{(b)} = [\mathbf{I}, \mathbf{P}]$ with $\mathbf{I}$ of size $K \times K$ and $\mathbf{P}$ of size $K\times (N-K)$. Let a vector $\bm q = (q_0,\cdots,q_{N-1})$ record the permuted positions such that $\widetilde{\mathbf{G}}^{(b)} = \mathbf{G}^{(b)} \mathbf{\Pi}_{\bm q}$, where $\mathbf{\Pi}_{\bm q}$ is a permutation matrix defined by $\bm q = (0,1,\cdots,N-1) \mathbf{\Pi}_{\bm q}$. Define $\mathcal{B}=\{q_0, q_1,\cdots, q_{K-1}\}\cap \mathcal{L}$ of size $|\mathcal{B}| = K-\Delta$. Without loss of generality, we may rewrite  $\widetilde{\mathbf{G}}^{(b)}$ in the following form by row and column permutations,
\begin{equation}\label{1GE}
	\widetilde{\mathbf{G}}_1^{(b)}=
	\begin{bNiceArray}{cw{c}{1.3cm}|[tikz=densely dashed]cw{c}{1.9cm}|[tikz=densely dashed]cw{c}{1.3cm}}[margin, first-row, last-col]
		\Block{1-2}{_{K-\Delta\textrm{ columns}}} & & \Block{1-2}{_{|\mathcal{L}|-K+\Delta\textrm{ columns}}} & & \Block{1-2}{_{N-|\mathcal{L}|\textrm{ columns}}} \\
		\Block{3-2}{\mathbf{I}_{|\mathcal{B}|}} & &\Block{3-2}{\mathbf{Q}_{11}} & &\Block{3-2}{\mathbf{Q}_{12}} & & \Block{3-1}{^{\rotate K-\Delta\textrm{ rows}}} \\
		& & & \\
		& & & \\
		\hdashline[3pt/1pt]
		%\hdashline[3pt/1pt]
		%	\vspace*{0.2cm}\hline
		\Block{1-2}{\mathbf{0}}& &\Block{1-2}{\mathbf{Q}_{21}} & & \Block{1-2}{\mathbf{Q}_{22}} & & \Block{1-1}{^{\rotate \Delta\textrm{ rows  }}} \\
	\end{bNiceArray},
\end{equation}
where the first $|\mathcal{L}|$ positions correspond to the positions in $\mathcal{L}$ and the columns of $\mathbf{I}_{|\mathcal{B}|}$ are associated with $\mathcal{B}$. If $\text{rank}(\mathbf{Q}_{21}) = \Delta$, which occurs with a probability\footnote{It can be proved that the probability is not less than $1-2^{-\delta}$ that a totally random binary matrix of size $\Delta \times (\Delta+\delta)$ is full-rank.} of about $1-2^{-(|\mathcal{L}|-K)}$, then the GE is applied to $\mathbf{Q}_{21}$,  referred to as the change-of-basis here, resulting in the following matrix
%where the column swapping only occurs in the $\Delta + \delta$ columns of $\mathbf{Q}_2$, 
\begin{equation}\label{2GE}
	\widetilde{\mathbf{G}}_2^{(b)}=
	\begin{bNiceArray}{cw{c}{1.2cm}|[tikz=densely dashed] cw{c}{1cm}cw{c}{1.4cm}|[tikz=densely dashed]cw{c}{1.4cm}}[margin, first-row, last-col]
		\Block{1-2}{_{K-\Delta\textrm{ columns}}} & & \Block{1-2}{_{\Delta\textrm{ columns}}} & & \Block{1-2}{_{|\mathcal{L}|-K\textrm{ columns}}}& & \Block{1-2}{_{N-|\mathcal{L}|\textrm{ columns}}} \\
		\Block{3-2}{\mathbf{I}_{|\mathcal{B}|}} & &\Block{3-2}{\mathbf{0}} & &\Block{3-2}{\mathbf{R}_{11}} & &\Block{3-2}{\mathbf{R}_{12}} & & \Block{3-1}{^{\rotate K-\Delta\textrm{ rows}}} \\
		& & & \\
		& & & \\
		\hdashline[3pt/1pt]
		%	\vspace*{0.2cm}
		%\hline
		\Block{1-2}{\mathbf{0}}& &\Block{1-2}{\mathbf{I}_{\Delta}} & &\Block{1-2}{\mathbf{R}_{21}} & & \Block{1-2}{\mathbf{R}_{22}} & & \Block{1-1}{^{\rotate \Delta\textrm{rows  }}} \\
	\end{bNiceArray}.
\end{equation}
If $\text{rank}(\mathbf{Q}_{21}) < \Delta$, we simply set the erased bits in the first $K$ positions as zero after the change-of-basis,  delivering a decoding output. For the case of $n-k < k$, the algorithm performs the Lagrange interpolation to derive a systematic parity-check matrix and performs the change-of-basis on the corresponding matrix.  It is more efficient for codes of rate $R > 1/2$. 

The cardinality $L \triangleq |\mathcal{L}|$ of the set $\mathcal{L}$ is a binomial random variable with the probability mass function~(PMF) given by
\begin{equation}
	P(L = \ell) = 
	\binom{N}{\ell}\bigl(1-\epsilon\bigr)^{\ell}
	\epsilon^{N - \ell}.
\end{equation}
Hence, the ML decoding performance of a generalized RS code $\mathscr{C}[N,K]_{2^m}$ with binary minimum Hamming distance $d_{\text{min}}$ over a BEC with erasure probability $\epsilon$ can be expressed as 
\begin{equation}
	\begin{aligned}\label{MLbound}
		P_{\text{ML}}
		&\leq\text{Pr}\left\{\text{the rank of the $K \times L$ matrix $\mathbf{G}^{(b)}_{\mathcal{L}}$ is less than $K$}\right\}\\
		& = \sum_{\ell=0}^{N} P(L=\ell)\cdot\text{Pr}\left\{\text{rank}\left(\mathbf{G}^{(b)}_{\mathcal{L}}\right)< K | L=\ell\right\}.
	\end{aligned}
\end{equation}
For $\ell < K$, the rank of $\mathbf{G}^{(b)}_{\mathcal{L}}$ is definitely less than $K$. When $L>N-d_{\text{min}}$, there must exist a one-to-one correspondence between the unerased positions of the received vector and the message vector, indicating that $\text{Pr}\left\{\text{rank}\left(\mathbf{G}^{(b)}_{\mathcal{L}}\right)< K | L>N-d_{\text{min}}\right\} = 0$. Thus, we have 

\begin{equation}
	\begin{aligned}\label{MLbound}
		P_{\text{ML}}
	&\leq P(L < K) + \sum_{\ell=K}^{N-d_{\text{min}}} P(L=\ell)\cdot\text{Pr}\left\{\text{rank}\left(\mathbf{G}^{(b)}_{\mathcal{L}}\right)< K | L=\ell\right\}\\
		&\overset{(*)}{\leq} P(L < K) + \sum_{\ell=K}^{N-(n-k+1)} P(L=\ell)\cdot\text{Pr}\left\{\text{rank}\left(\mathbf{G}^{(b)}_{\mathcal{L}}\right)< K | L=\ell\right\},
	\end{aligned}
\end{equation}
where $P(L < K) = \sum\limits_{\ell=0}^{K-1}	\binom{N}{\ell}\bigl(1-\epsilon\bigr)^{\ell}
\epsilon^{N - \ell}$, and the inequality $(*)$ holds since $d_{\text{min}} \geq n-k+1$.

Recall that the probability of a $K\times \ell$ totally random matrix being not full-rank is less than $2^{-(\ell-K)}$. 
Taking into account that  $\mathbf{G}^{(b)}_{\mathcal{L}}$ is not totally random, we have the following approximate upper bound
\begin{equation}
	\text{Pr}\left\{\text{rank}\left(\mathbf{G}^{(b)}_{\mathcal{L}}\right)<K | L=\ell\right\} \lesssim 2^{-(\ell-K)},
\end{equation}
for $\ell \geq K$.

Therefore, the average ML decoding performance of the generalized RS code can be estimated as follows, referred to as Approx-UB,
\begin{equation}
	\begin{aligned}\label{approMLbound}
		P_{\text{ML}} &\lesssim  P(L < K) + \sum_{\ell=K}^{N-(n-k+1)} P(L=\ell)\cdot 2^{-(\ell-K)}.
	\end{aligned}
\end{equation}

As we can see, the Lagrange interpolation can be completed in parallel by two time steps,
% assuming there are no resource limitations and all parallelizable instructions are executed in a single step, 
resulting in negligible decoding latency. 
In the GE process, suppose that reducing a column to a unit vector constitutes one iteration. The original GE begins with a given offline systematic matrix, taking on average about $K \varepsilon$ iterations. In contrast, the change-of-basis after the Lagrange interpolation reduces the average number of iterations down to typically less than $K \varepsilon$. This reduction can potentially  decrease the decoding latency.

% present the performance\\

\begin{example}
	Consider the generalized RS codes with different code rates over BECs with $\varepsilon \in \{0.1,0.2,0.3\}$. The codes are defined over the field $\mathbb{F}_{2^8}$ with $m=8$ and  mapped into $\mathbb{F}_{2}^{1024}$, denoted by $\mathscr{C}[1024,K]_{2^8}$ with $K \in \{256,512,768\}$. We have compared  the average number of iterations for the change-of-basis and the original GE, as shown in Table.~\ref{tab1}. These results are counted averaging over one million decoding trials. It can be seen that the average number of iterations increases as the code rate increases. Compared with the original GE approach, on average, the change-of-basis reduces the number of iterations across all codes by over $30\%$. 
	%Specifically, for the rate $1/2$ code, the LC-OSD achieves a reduction of $68\%$ in the number of iterations.
	
	%By selecting the basis that hasthe most common elements with M, the number of iterations required by the IBU can be reduced. 
	
\end{example}
\begin{table}[!t]
	%\begin{center}
	\renewcommand{\arraystretch}{1.3}
	\centering
	\caption{The average number of iterations for $\mathscr{C}_{\text{GRS}}[1024,K]_{2^8}$. Here, $m = 8$, $K \in \{256,512,786\}$, and $\varepsilon \in \{0.1,0.2,0.3\}$.}
	\label{tab1}
	\begin{tabular}{ |c|c | c|  c |c | c|  c|c | c|  c|}
		\hline
		$\varepsilon$ & \multicolumn{3}{c|}{$0.1$} &\multicolumn{3}{c|}{$0.2$} &\multicolumn{3}{c|}{$0.3$}  \\
		\hline
		$K$ & $256$ & $512$ &  $768$  & $256$ & $512$ &  $768$  & $256$ & $512$ &  $768$ \\
		\hline 
		Original GE &  $34.89$ & $51.75$ &  $76.80$  & $51.27$ & $102.41$ &  $153.61$  & $76.80$ & $153.60$ &  $188.75$ \\
		\hline 
	   Change-of-basis & $2.44$ & $9.69$ &  $40.97$  & $10.67$ & $44.58$ &  $106.33$  & $26.26$ & $88.28$ &  $138.55$ \\
	  \hline
	\end{tabular}
	%\end{center}
\end{table}

\begin{example}
We have simulated the generalized RS codes with different code lengths over the field $\mathbb{F}_{2^8}$ and $\mathbb{F}_{2^9}$. The simulation results are shown in Figs.~\ref{BECfer1}-\ref{BECfer2}. For comparison, we have also plotted in Fig.~\ref{BECfer1} the lower bound under the successive cancellation~(SC) decoding~\cite{ref1Arikan} for the polar codes with a code rate of $1/2$, and the FER performance of the polar code with a total code rate of $1/2$ and a cyclic redundancy check~(CRC) of length $6$, using an ML decoder~\cite{EML2021}. As we can see, the generalized RS codes perform well when the code length is large, providing numerical evidence that the generalized RS code achieves channel capacity as the code length increases. Additionally, the generalized RS codes outperform the polar codes with the same code length and dimension, under the ML decoding. The Approx-UB~\eqref{approMLbound} is also provided in Figs.~\ref{BECfer1}-\ref{BECfer2}, from which we observe that the Approx-UB becomes tighter as the code length increases, particularly in the low erasure probability region. %The Approx-UB~\eqref{approMLbound} is also provided in Figs.~\ref{BECfer1}-\ref{BECfer2}. We observe that the Approx-UB becomes tighter as the code length increases, particularly in the low erasure probability region.
\end{example}

\begin{figure}[!t]
	\centering
	\includegraphics[width=4.8in]{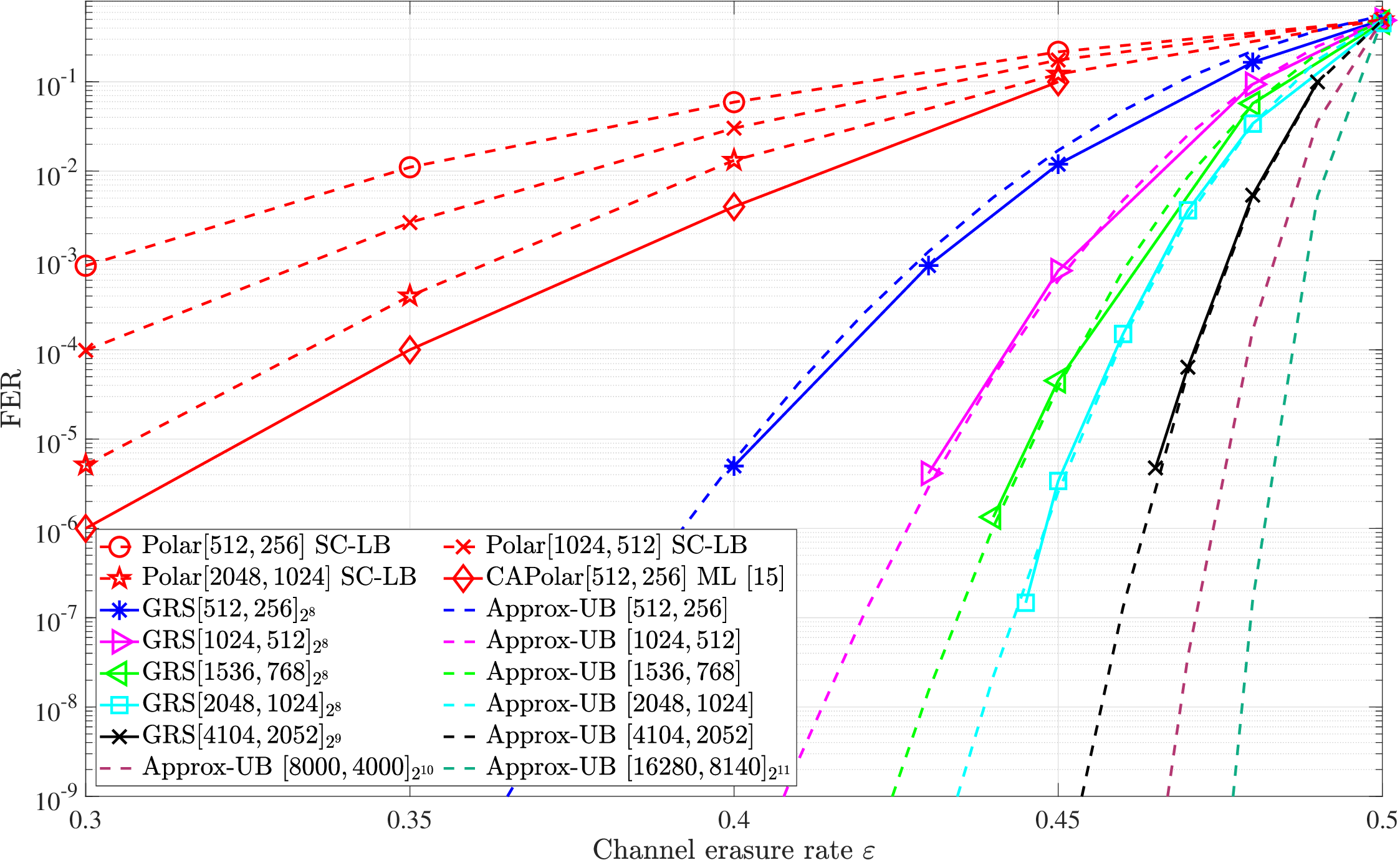}
	\caption{Performance comparison between rate-$0.5$ generalized RS codes and rate-$0.5$ polar codes with different code lengths over BECs.}
	\label{BECfer1} 
\end{figure}

\begin{figure}[!t]
	\centering
	\includegraphics[width=4.8in]{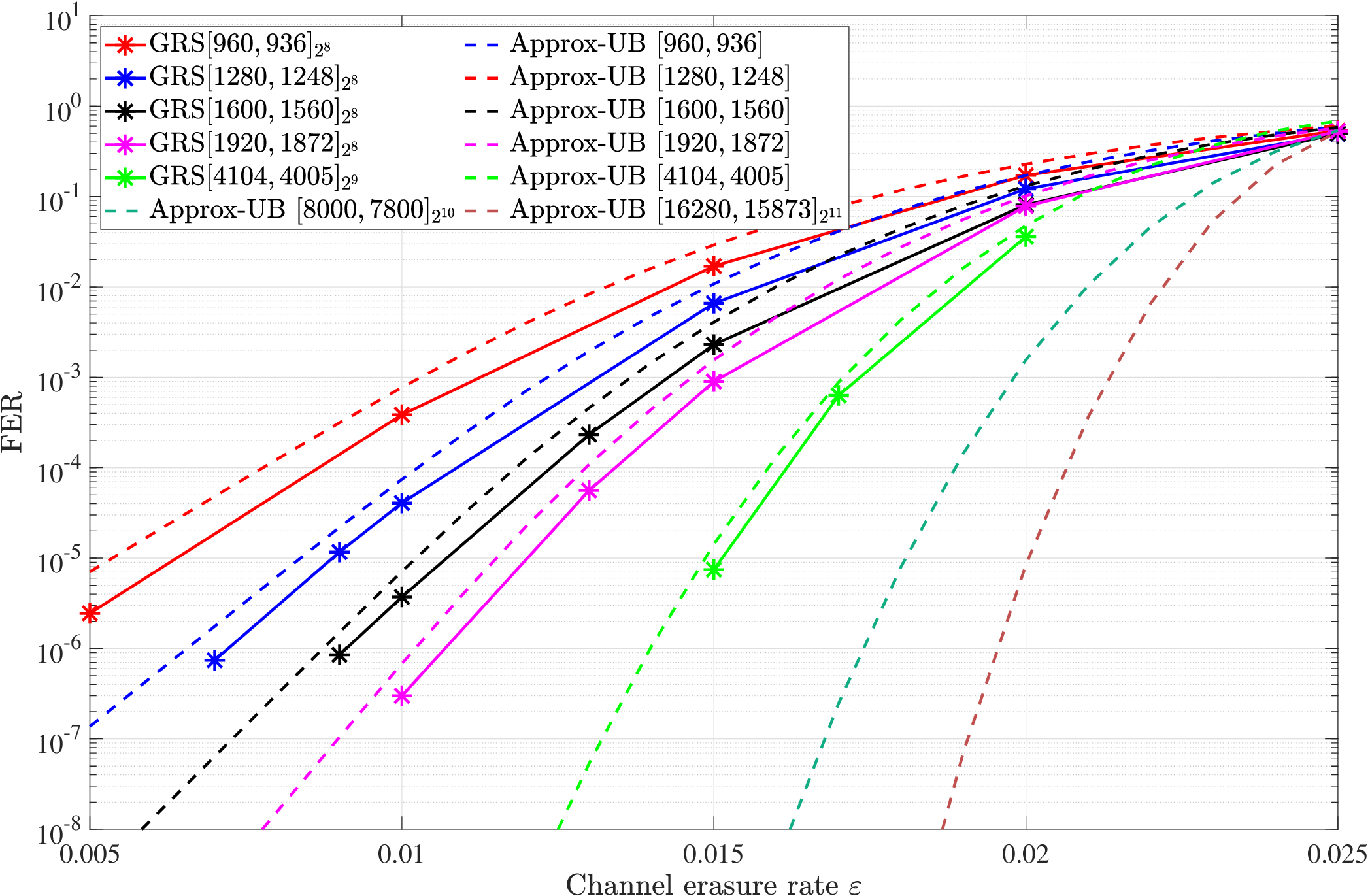}
	\caption{Performance of rate-$0.975$ generalized RS codes with different code lengths over BECs. Here, the curves are shown for $N = 960, 1280, 1600, 1920, 4104$.}
	\label{BECfer2} 
\end{figure}

%\begin{figure}[!t] %\label{Fig_OESD}
%	\centering
%	\subfloat[Performance. Here, from left to right, the curves are marked by $(K)$ with the dimension $K$ inside.\label{fig:rmfer}]{\includegraphics[width=3.1in]{fig/R0_5.eps}}
%	\hfill
%	\subfloat[Average number of queries at the target FER~$10^{-5}$.\label{fig:rmQuery}]{\includegraphics[width=3.1in]{fig/R0_975.eps}}
%	\caption{
%		ML decoding of the RM codes $\mathscr{C}_{\rm{RM}}[64,K]$ with OSD~\cite{DorshOSD} and GCD~($L=1$).
%	}
%	\label{BECfer}
%\end{figure}

As we know, both original RS codes and generalized RS codes are MDS codes. Specifically, these codes defined over $\mathbb{F}_{p^m}$ with (in terms of symbols) code length $n$ and dimension $k$ have a minimum Hamming distance of $n-k+1$. By construction, the $p$-ary images of both original RS codes and generalized RS codes maintain $p$-ary minimum distance of at least $(N-K)/m+1$. Given a $p$-ary code length $N$ and dimension $K$, as $m$ increases, both original RS codes and generalized RS codes using symbol-level decoding algorithms have a significant degradation in performance over the $p$-ary transmission channels due to the decrease in minimum Hamming distance. However, this may not be the case for $p$-ary-level decoding algorithms. 

Consider the BEC channels again. As illustrated in the following example, the performance of original RS codes degrades as $m$ increases, while the performance of generalized RS codes remains stable, both using bit-level decoding. This is reasonable as explained below. Indeed, in terms of bits, the minimum Hamming distance of generalized RS codes can be greater than that of original RS codes due to the introduced non-zero scaling vector ($\alpha^{j_0}, \cdots, \alpha^{j_{n-1}}$), resulting in potential performance gain in the large fields compared with original RS codes.

\begin{example}
	Consider the RS codes whose binary images have identical code length $N$ and dimension $K$ over different fields~(different $m$). We present in Fig.~\ref{RSGRS_differentm} the performance of the original RS codes and the generalized RS codes over BECs, from which we observe that as $m$ increases from $5$ to $20$, the performance loss of the original RS codes becomes significant. In contrast, the performance of the generalized RS codes over different $m$ is similar. This example also serves to explain why we need generalized RS codes for proving the coding theorem. Of course, to construct a code of length $N$ and $K$, we prefer small $m$ for reducing the complexity of Lagrange interpolation. 
\end{example}

\begin{figure}[!t]
	\centering
	\subfloat[Original RS codes.\label{RS_differentm}]{\includegraphics[width=3.1in]{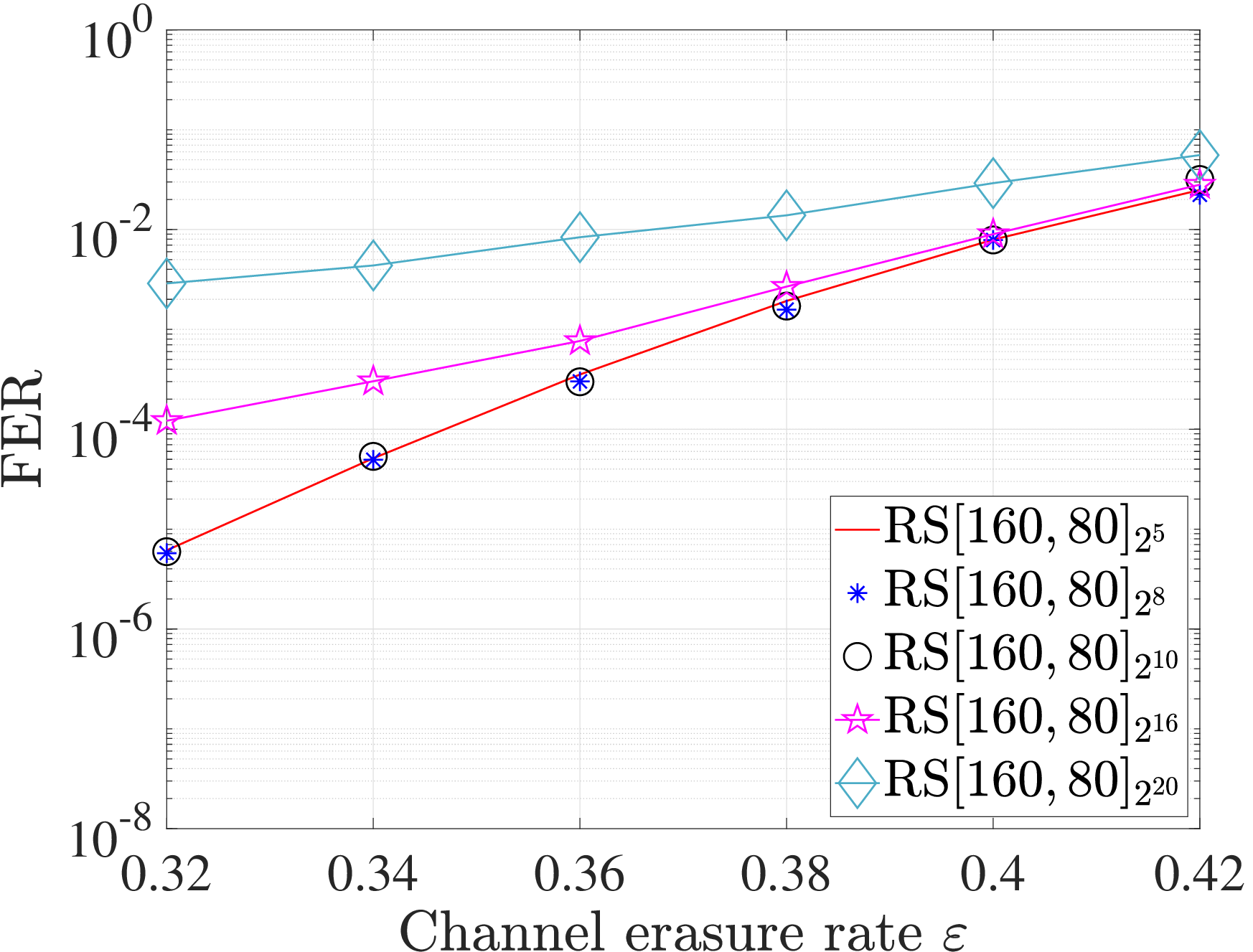}}
	\hfill
	\subfloat[Generalized RS codes.\label{GRS_differentm}]{\includegraphics[width=3.1in]{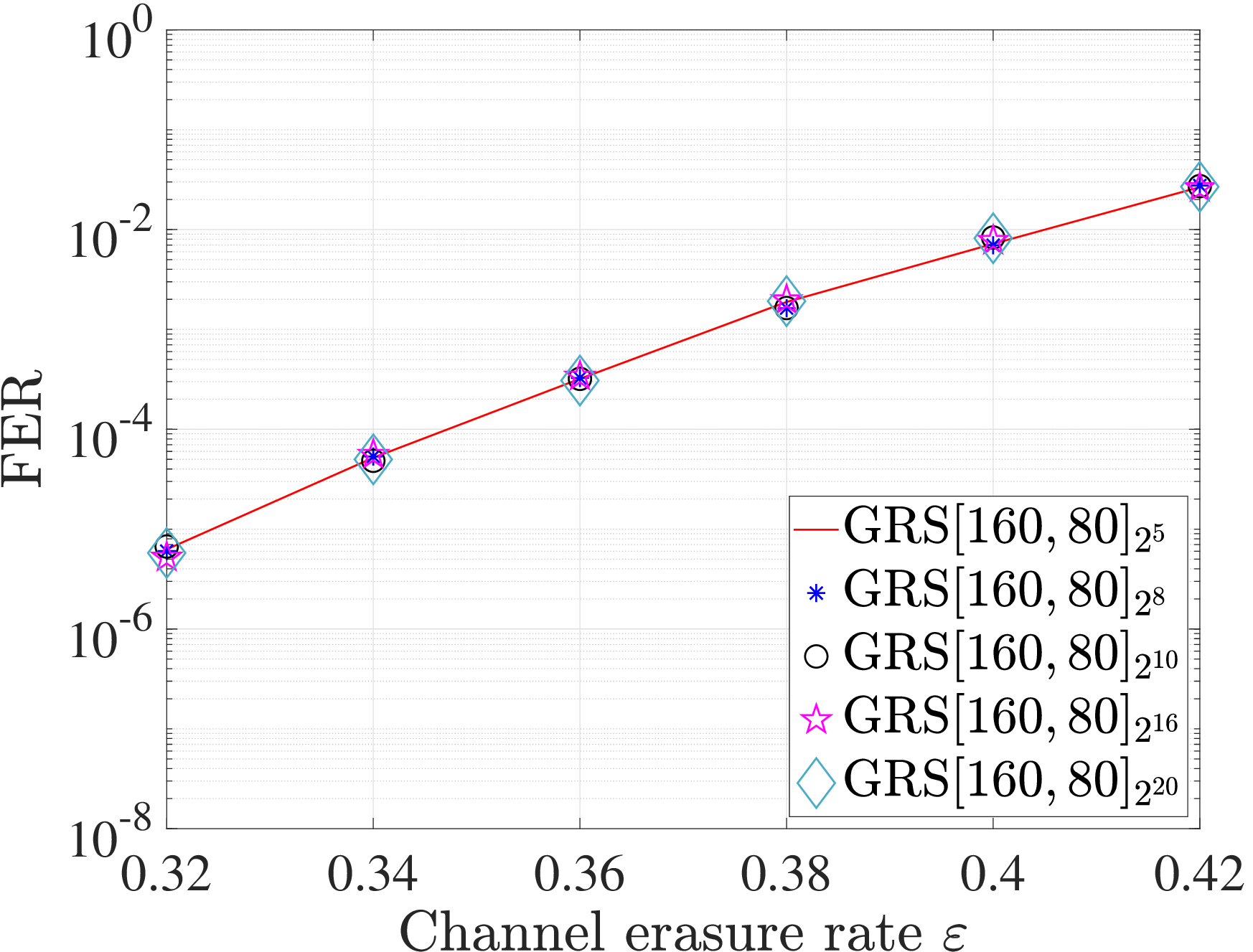}}
	\caption{
			Performance of the original RS codes and the generalized RS codes  whose binary images have same code length $N$ and dimension $K$ over different fields~(different $m$) over BECs.	}
	\label{RSGRS_differentm}  
\end{figure}

\subsection{BPSK-AWGN Channels}
When incorporating soft information to measure reliability, it becomes clear that the proposed decoding algorithm developed for the BECs can also be adapted to BPSK-AWGN channels.

Consider the vector $\bm x = \phi^{-1}(\bm c) \in \mathbb{F}_2^N$, which is transmitted over BPSK-AWGN channels. Upon receiving $\boldsymbol{y} \in \mathbb{R}^N$, we calculate  the log-likelihood ratio~(LLR) vector $\boldsymbol{r}$ by 
\begin{equation}\label{llr}
	r_i = \log {\frac{P_{Y|X}(y_i|x_i=0)}{P_{Y|X}(y_i|x_i = 1)}},\ 0 \leq i < N,
\end{equation}
and the hard-decision vector $\boldsymbol{z}$ by 
\begin{equation}\label{harddecison}
	z_i=
	\begin{cases} 
		0, & \mbox{if }r_i\geq0\\
		1, & \mbox{if }r_i<0\\
	\end{cases}
	,~0\leq i < N.
\end{equation}
We refer to the absolute values of the LLRs as the bit reliabilities corresponding to the hard decisions. The hard-decision bit sequence $\boldsymbol{z}$ is also written as a hard-decision symbol sequence $\boldsymbol{\zeta}=(\zeta_0, \zeta_1, \dots, \zeta_{n-1}) = \phi(\bm z)$ with $\zeta_i \in \mathbb{F}_q$, $0\leq i < n$. By sorting the LLR vector $\boldsymbol{r}$, we can find the most $K$ reliable bits in $\boldsymbol{z}$, whose indices are denoted collectively by $\mathcal{Z}^*$. Then, for two components $\zeta_i$ and $\zeta_j$ of the hard-decision symbol sequence $\boldsymbol{\zeta}$, we say that $\zeta_i$ is more reliable than $\zeta_j$ if $\zeta_i$ contains more bits in $\mathcal{Z}^*$ than $\zeta_j$ does. In the case when two symbols contain the same number of bits in $\mathcal{Z}^*$, the symbol with a greater sum of bit reliabilities is considered more reliable. Similar to the decoding process described in~Sec.~IV.~A, we then employ the parallel Lagrange interpolation and the change-of-basis to derive an ordered systematic generator matrix. Given the binary image of the systematic generator matrix, where the identity matrix corresponds to the MRB, we can apply the LC-OSD~\cite{LC_OSD2022,GESTCOM} for the remaining decoding process.

\begin{example}
	We have simulated the generalized RS codes with different code rates over BPSK-AWGN channels. The simulated codes are defined over the field $\mathbb{F}_{2^8}$ with length $n=16$ and  mapped into $\mathbb{F}_{2}^{128}$, denoted by $\mathscr{C}[128,K]_{2^8}$ with $K$ as the dimension. They are decoded by the LC-OSD~\cite{LC_OSD2022}~\cite{GESTCOM}. As shown in Fig.~\ref{cap_rcu}, we compare the generalized RS codes with the RCU bounds~\cite{polyanskiy2010channel}, from which we can see that the generalized RS codes perform well, closely approaching the RCU bounds. This suggests that the generalized RS codes are competitive in the finite-length region. To strength this point, we have compared the $\mathscr{C}_{\text{GRS}}[128, 64]_{2^8}$ and the extended Bose-Chaudhuri-Hocquenghem (eBCH) code $\mathscr{C}_{\text{eBCH}}[128, 64]$. The simulation results are shown in Fig.~\ref{eBCH}, from which we observe that the generalized RS code exhibits similar
	decoding performance to the eBCH code. However, the LC-OSD of generalized RS codes can have low latency due to the parallel Lagrange interpolation for constructing the extended MRB. Furthermore, compared with many other good short-length codes, such as eBCH codes, the generalized RS codes are compelling since they are rate-flexible, essentially spanning the whole range of code rates. Also notice that, as multiple-rate codes, the generalized RS codes can be concatenated with classified enumerative coding to implement joint source-channel coding~(JSCC)~\cite{QuasiOSD} for short packet transmission.
\end{example}

\begin{figure}[!t]
	\centering
	\includegraphics[width=3.3in]{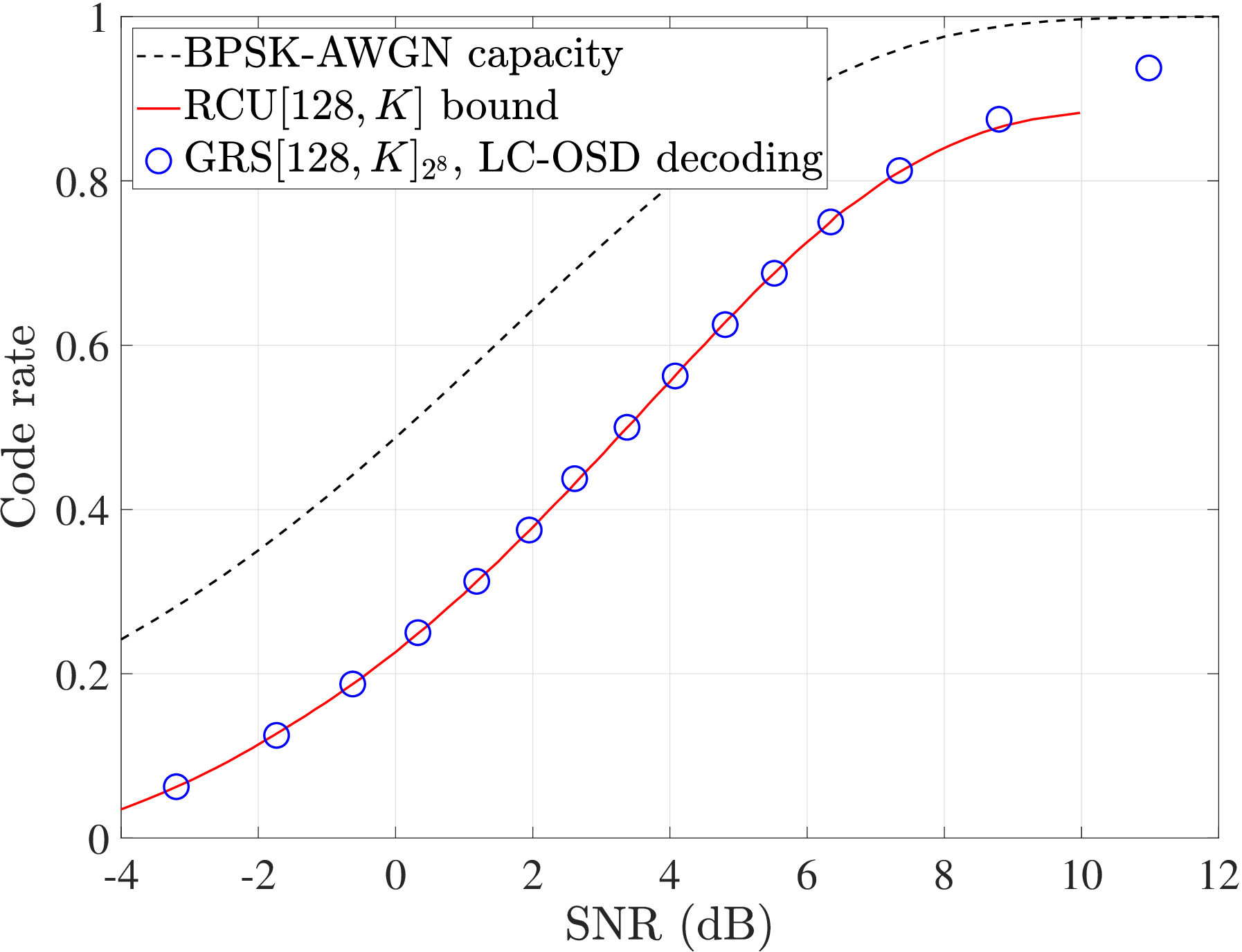}
	\caption{The finite-length capacity and the required SNR by the binary images of generalized RS codes $\mathscr{C}[128,K]_{2^8}$ under the LC-OSD~\cite{GESTCOM} at the target FER of $10^{-5}$. Here, the LC-OSD is with the  dynamic approximate ideal (DAI) stopping criterion~\cite{LC_OSDljf2023} and decoding parameters $\delta=10$, $\ell_{\text{max}}=2^{16}$ for $K\in\{8k, 1\leq k \leq 15\}$.}
	\label{cap_rcu} 
\end{figure}

\begin{figure}[!t]
	\centering
	\includegraphics[width=3.5in]{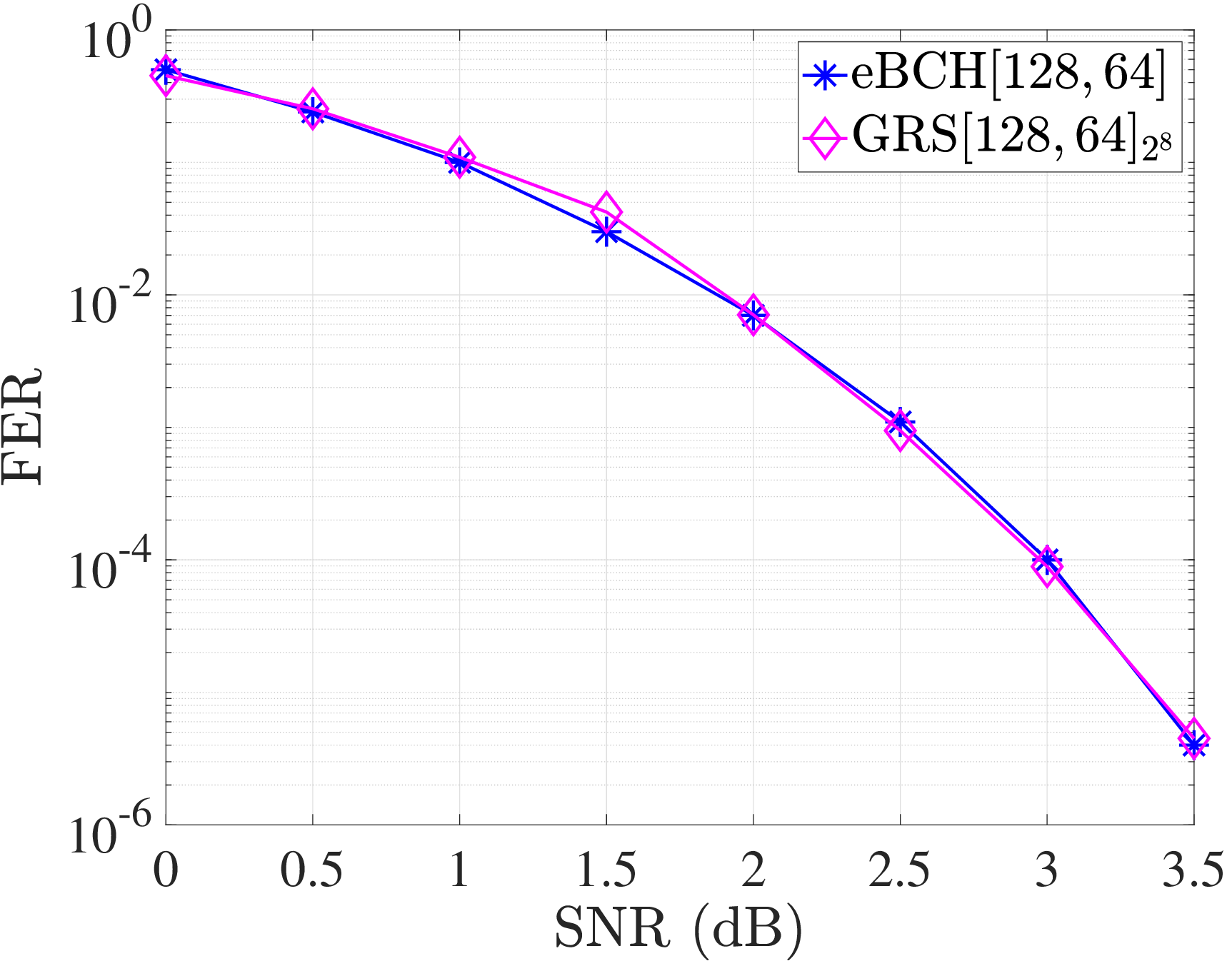}
	\caption{Performance of the eBCH code $\mathscr{C}_{\text{eBCH}}[128,64]$ and the generalized RS code $\mathscr{C}_{\text{GRS}}[128,64]_{2^8}$ under LC-OSD. Here, the LC-OSD is with the  DAI stopping criterion~\cite{LC_OSDljf2023} and decoding parameters $\delta=10$, $\ell_{\text{max}}=2^{16}$.}
	\label{eBCH} 
\end{figure}

%\begin{figure}[!t]
%	\centering
%	\subfloat[The finite-length capacity and the required SNR with the LC-OSD. Here, the target FER is $10^{-5}$.\label{cap_rcu}]{\includegraphics[width=0.37\textwidth]{fig/GRS_RCU_cap.eps}}
%	\\
%	\subfloat[Performance bounds. Here, $k \in \{7,15,23,27\}$.
%	\label{improveUB}]{\includegraphics[width=0.37\textwidth]{fig/IUB_fer.eps}}
%	\caption{
%		Simulation results of the generalized RS codes $\mathscr{C}[31,k]_{2^5}$, which are mapped into $\mathbb{F}_{2}^{155}$ for LC-OSD. Here, the LC-OSD is with the  dynamic approximate ideal (DAI) stopping criterion~\cite{LC_OSDljf2023} and decoding parameters $\delta=8$, $\ell_{\text{max}}=2^{14}$ for $k\in\{5,19,21,23,25,27\}$, and $\ell_{\text{max}}=2^{16}$ for $ k \in \{7,9,11,13,15,17\}$.	}
%	\label{fer_LCOSD}  
%\end{figure}

\subsection{3PAM-AWGN Channels}

In many real-world scenarios, where information sources generate three distinct states or symbols, such as sensor networks measuring environmental conditions at three levels: low, medium, and high, ternary coding is a natural choice. Consider a ternary source vector $\bm v \in \mathbb{F}_3^K$. Recall the system model shown in Fig.~\ref{system_model}. Associated with a vector $\bm u = \phi(\bm v) \in \mathbb{F}_q^k$ with  $q=3^m$ is a generalized RS codeword $\bm c = \bm u\mathbf{G} \in \mathbb{F}_q^n$.  Then the coded sequence $\bm x = \phi^{-1}(\bm c + \bm a) \in \mathbb{F}_3^N$  is mapped into 3PAM signals, and transmitted over AWGN channels, resulting in  $\boldsymbol{y} \in \mathbb{R}^N$. Let $\bm a^{(p)} = \phi^{-1}(\bm a) \in \mathbb{F}_3^N$.

Upon receiving $\boldsymbol{y}$, we can calculate the hard-decision vector $\boldsymbol{z} \in \mathbb{F}_3^N$ by 
\begin{equation}\label{harddecison-3PAM}
	z_i= {\arg\max}_{\beta \in \mathbb{F}_3}P_{Y|X}\left(y_i|x_i=\beta+\bm a^{(p)}[i]\right),~0\leq i < N.
\end{equation}
The ternary digit~(trit) reliabilities $\bm r \in \mathbb{R}^N$ corresponding to the hard decisions are defined as
\begin{equation}\label{r-3PAM}
	r_i=\log {\frac{P_{Y|X}(y_i|x_i=z_i+\bm a^{(p)}[i])}{  {\max}_{\beta \in \mathbb{F}_3, \beta\neq z_i} P_{Y|X}\left(y_i|x_i = \beta+\bm a^{(p)}[i]\right)}},~0\leq i < N,
\end{equation}
which is the log-likelihood ratio of the hard~(the first reliable) decision to the second reliable decision. The hard-decision trit sequence $\boldsymbol{z}$ is also written as a hard-decision symbol sequence $\boldsymbol{\zeta}=(\zeta_0, \zeta_1, \dots, \zeta_{n-1}) = \phi(\bm z)$ with $\zeta_i \in \mathbb{F}_q$, $0\leq i < n$. By sorting the vector $\boldsymbol{r}$, we can find the most $K$ reliable trits in $\boldsymbol{z}$, whose indices are denoted collectively by $\mathcal{Z}^*$. Then, for two components $\zeta_i$ and $\zeta_j$ of the hard-decision symbol sequence $\boldsymbol{\zeta}$, we say that $\zeta_i$ is more reliable than $\zeta_j$ if $\zeta_i$ contains more trits in $\mathcal{Z}^*$ than $\zeta_j$ does. In the case when two symbols contain the same number of trits in $\mathcal{Z}^*$, the symbol with a greater sum of trit reliabilities is considered more reliable. Then we can employ the parallel Lagrange interpolation and the change-of-basis to derive the systematic generator matrix. Correspondingly, the systematic generator matrix can be written as a ternary systematic matrix $\widetilde{\mathbf{G}}^{(p)} = [\mathbf{I}, \mathbf{P}]$ with $\mathbf{I}$ of size $K \times K$ and $\mathbf{P}$ of size $K\times (N-K)$.  Let a vector $\bm p = (p_0,\cdots,p_{N-1})$ record the permuted positions such that $\widetilde{\mathbf{G}}^{(p)} = \mathbf{G}^{(p)} \mathbf{\Pi}_{\bm p}$, where $\mathbf{\Pi}_{\bm p}$ is a permutation matrix defined by $\bm p = (0,1,\cdots,N-1) \mathbf{\Pi}_{\bm p}$.

For a test vector $\hat{\boldsymbol{v}} \in \mathbb{F}_3^{N}$, we can define its corresponding test error pattern~(TEP) $\boldsymbol{e}\in \mathbb{F}_3^N$ as
\begin{equation}
	\boldsymbol{e} \triangleq \boldsymbol{z}-\hat{\boldsymbol{v}}.
\end{equation}
This can be written as $\boldsymbol{z}= \hat{\boldsymbol{v}}+\boldsymbol{e}$ and hence the channel is transformed into an additive noise channel, which accepts the codeword as input and delivers the hard-decision vector as output. Define the soft weight of a TEP $\boldsymbol{e}$, denoted by $\gamma(\boldsymbol{e})$, as
\begin{equation}\label{softweight}
\begin{aligned}	
	\gamma(\boldsymbol{e}) &\triangleq \log  \frac{P_{Y|X}(\boldsymbol{y}|\boldsymbol{z} + \bm a^{(p)})}{P_{Y|X}(\boldsymbol{y}|\boldsymbol{z}-\boldsymbol{e}+ \bm a^{(p)})} \\
	&= \sum\limits_{i:e_i \neq 0} \log P_{Y|X}(y_i|z_i + \bm a^{(p)}[i])-\log P_{Y|X}(y_i|z_i - e_i + \bm a^{(p)}[i])\\
	& =  \sum\limits_{i:e_i \neq 0} \gamma_i(e_i),
\end{aligned}	
\end{equation}
where $\gamma_i(e_i) =\log P_{Y|X}(y_i|z_i + \bm a^{(p)}[i])-\log P_{Y|X}(y_i|z_i - e_i + \bm a^{(p)}[i]) $. We can see that $\gamma(\boldsymbol{e}) \geq 0$ for any vector $\boldsymbol{e} \in \mathbb{F}_3^{N}$ and $\gamma(\boldsymbol{e}) = 0$ for $\boldsymbol{e} = \boldsymbol{0}$. In contrast to the Hamming weight, the soft weight of $\boldsymbol{e}$, as a weighted sum, is determined not only by its non-zero components but also by the corresponding reliabilities $\gamma_i(e_i)$ with $e_i \neq 0$.  Then the ML decoding is equivalent to the lightest soft weight decoding. That is,  the ML decoding is equivalent to
\begin{equation}
	\begin{aligned}
		\min\limits_{\boldsymbol{e}\in \mathbb{F}_3^N} \quad &\gamma(\boldsymbol{e})\\
		\mbox{s.t.}\quad & \boldsymbol{e}\mathbf{H}^T=\boldsymbol{s},
	\end{aligned}
\end{equation}
where $\boldsymbol{s}=\boldsymbol{z}\mathbf{H}^T$ is the computable syndrome and $\mathbf{H}$ is the parity-check matrix of size $(N-K) \times N$. Notice that the ternary parity-check matrix $\mathbf{H}$ can be derived  from the ternary generator matrix $\mathbf{G}^{(p)}$.

Given a preset non-negative integer $\delta \leq N-K$, the derived ternary systematic matrix $\widetilde{\mathbf{G}}^{(p)} = [\mathbf{I}, \mathbf{P}]$ can be written as follows,
\begin{equation}
	\widetilde{\mathbf{G}}^{(p)}=
	\begin{bNiceArray}{cw{c}{1cm}|[tikz=densely dashed]cw{c}{1.3cm}|[tikz=densely dashed]cw{c}{1.3cm}}[margin, first-row, last-col]
		\Block{1-2}{_{K\textrm{ columns}}} & & \Block{1-2}{_{\delta\textrm{ columns}}} & &\Block{1-2}{_{N-K-\delta\textrm{ columns}}}\\
		\Block{2-2}{\mathbf{I}} & &\Block{2-2}{\mathbf{P}_1} & & \Block{2-2}{\mathbf{P}_2} & &\Block{2-1}{^{\rotate K\textrm{ rows}}}\\
		& & &\\
	\end{bNiceArray}.
\end{equation}
Then a TEP $\widetilde{\boldsymbol{e}} = \bm e \mathbf{\Pi}_{\bm p}$ can be written as $\widetilde{\boldsymbol{e}}= (\widetilde{\boldsymbol{e}}_I,\widetilde{\boldsymbol{e}}_{P_1},\widetilde{\boldsymbol{e}}_{P_2})$ with $\widetilde{\boldsymbol{e}}_I \in \mathbb{F}_3^{K}$, $\widetilde{\boldsymbol{e}}_{P_1} \in \mathbb{F}_3^{\delta}$ and $\widetilde{\boldsymbol{e}}_{P_2} \in \mathbb{F}_3^{N-K-\delta}$. Similarly, $\widetilde{\boldsymbol{z}}= (\widetilde{\boldsymbol{z}}_I,\widetilde{\boldsymbol{z}}_{P_1},\widetilde{\boldsymbol{z}}_{P_2})$. By a valid TEP $\widetilde{\boldsymbol{e}}$, we mean a vector $\widetilde{\boldsymbol{e}}$ such that $(\widetilde{\boldsymbol{z}}-\widetilde{\boldsymbol{e}})\mathbf{\Pi}_{\bm p}^{-1}$ is a codeword. For any valid TEP $\widetilde{\boldsymbol{e}}$, we have
\begin{equation}
	\begin{aligned}\label{re-encoding}
		{\widetilde{\boldsymbol{e}}}_{I}\mathbf{P}_2^{\top}  + {\widetilde{\boldsymbol{e}}}_{P_2}  =
		{\widetilde{\boldsymbol{z}}}_{I}\mathbf{P}_{2}^{\top} + {\widetilde{\boldsymbol{z}}}_{P_2}.
	\end{aligned}
\end{equation}
\begin{equation}
	\label{equ:e-parity}
	\begin{aligned}
		{\widetilde{\boldsymbol{e}}}_{I}\mathbf{P}_1^{\top}  + {\widetilde{\boldsymbol{e}}}_{P_1}  =
		{\widetilde{\boldsymbol{z}}}_{I}\mathbf{P}_{1}^{\top} + {\widetilde{\boldsymbol{z}}}_{P_1},
	\end{aligned}
\end{equation}
From~\eqref{re-encoding}, we can see that ${\widetilde{\boldsymbol{e}}}_{P_2}$ and ${\widetilde{\boldsymbol{e}}}$ are uniquely determined by $(\widetilde{\boldsymbol{e}}_{I},\widetilde{\boldsymbol{e}}_{P_1})$. Hence, we only need to query $(\widetilde{\boldsymbol{e}}_{I},\widetilde{\boldsymbol{e}}_{P_1})$ instead of ${\widetilde{\boldsymbol{e}}}$. Utilizing the local constraints~\eqref{equ:e-parity}, we can query  $(\widetilde{\boldsymbol{e}}_{I},\widetilde{\boldsymbol{e}}_{P_1})$ in such an order that the partial soft weights are non-decreasing, i.e., $ \gamma\left(\widetilde{\boldsymbol{e}}_{I}^{(1)}\right)+\gamma\left(\widetilde{\boldsymbol{e}}_{P_1}^{(1)}\right) \leq  \cdots \leq \gamma\left(\widetilde{\boldsymbol{e}}_{I}^{(\ell_{\text{max}})}\right)+\gamma\left(\widetilde{\boldsymbol{e}}_{P_1}^{(\ell_{\text{max}})}\right)$, where  $\ell_{\text{max}}$ is the  maximum query number. This can be achieved by the serial list Viterbi algorithm~(SLVA)~\cite{Nambirajan1994SLVA} over a trellis which is specified by the matrix $[\mathbf{I},\mathbf{P}_{1}]$ and has at most $3^\delta$ states. Finally, the ternary decoding is terminated whenever the maximum query number is reached or the ML codeword has been identified, delivering the most likely candidate codeword as output. This query process is similar to the LC-OSD~\cite{LC_OSD2022,GESTCOM} methods for binary block codes but it extends LC-OSD into the realm of ternary decoding. Notice that, if $\delta = 0$, we can directly query $\widetilde{\boldsymbol{e}}_{I}$ in non-decreasing soft weights, with the aid of the flipping pattern tree~(FPT)~\cite{FPT}.

It is not difficult to imagine that a ternary source data of length $K$ can also be represented~(in theory) as a binary sequence of length $K\log_2 3$, allowing for the use of binary coding.  However, we will demonstrate in the following example that ternary coding can be preferable to binary coding for ternary sources.

\begin{example}
	We have simulated the ternary images of generalized RS codes $\mathscr{C}[63,K]_{3^3}$ with $K \in \{32, 35\}$ over 3PAM-AWGN channels. For a fair comparison, we have also simulated the binary images of generalized RS codes $\mathscr{C}[64,K']_{2^4}$ with $K' \in \{51, 56\}$ over BPSK-AWGN channels, which have similar code lengths and rates~(bits per channel use, denoted by $b_c$) to the ternary coding. The simulation results are shown in Fig.~\ref{fig_Ternary_source}, from which we can see that for comparable bits per channel use and code lengths, the ternary coding outperforms the binary coding. Specifically, the ternary generalized RS codes with $b_c = 0.80$ and $0.88$ achieve about $0.7$~dB and $1.5$~dB gain at $\rm{FER}=10^{-4}$, respectively, compared with the corresponding binary generalized RS codes. To further demonstrate the superiority of the ternary generalized RS codes over the binary codes, we have compared the ternary generalized RS code $\mathscr{C}_{\text{GRS}}[128,76]_{3^4}$ with the binary extended Hamming code $\mathscr{C}_{\text{Hamming}}[128,120]$.  The simulation results are shown in Fig.~\ref{TernaryHamming}, from which we see that the ternary generalized RS code achieves about  $1.2$ dB gain at $\rm{FER}=10^{-4}$ over the extended Hamming code.
\end{example}

\begin{figure}[!t]
	\centering
	\subfloat[\label{performance_51}]{\includegraphics[width=3.1in]{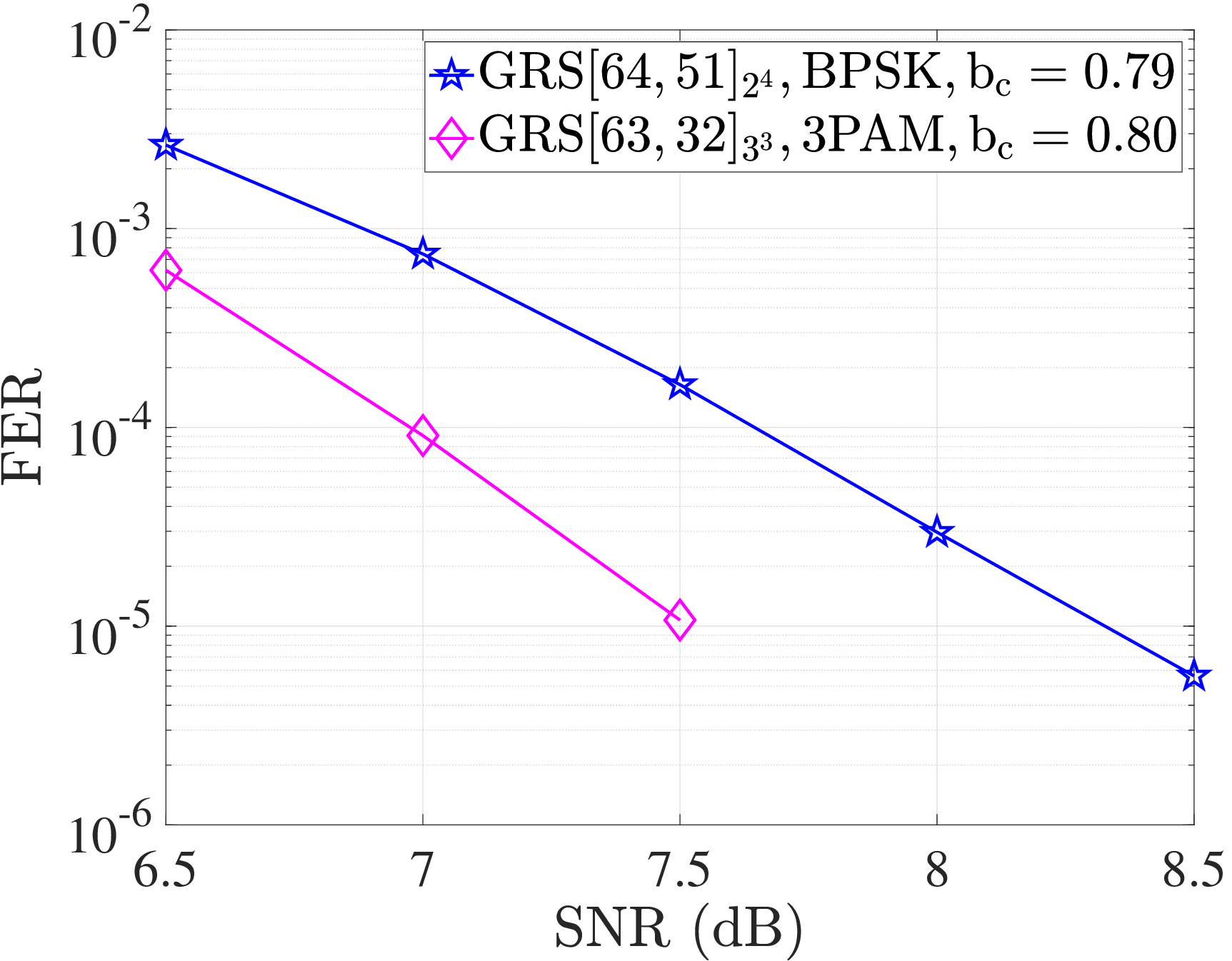}}
	\hfill
	\subfloat[\label{performance_56}]{\includegraphics[width=3.1in]{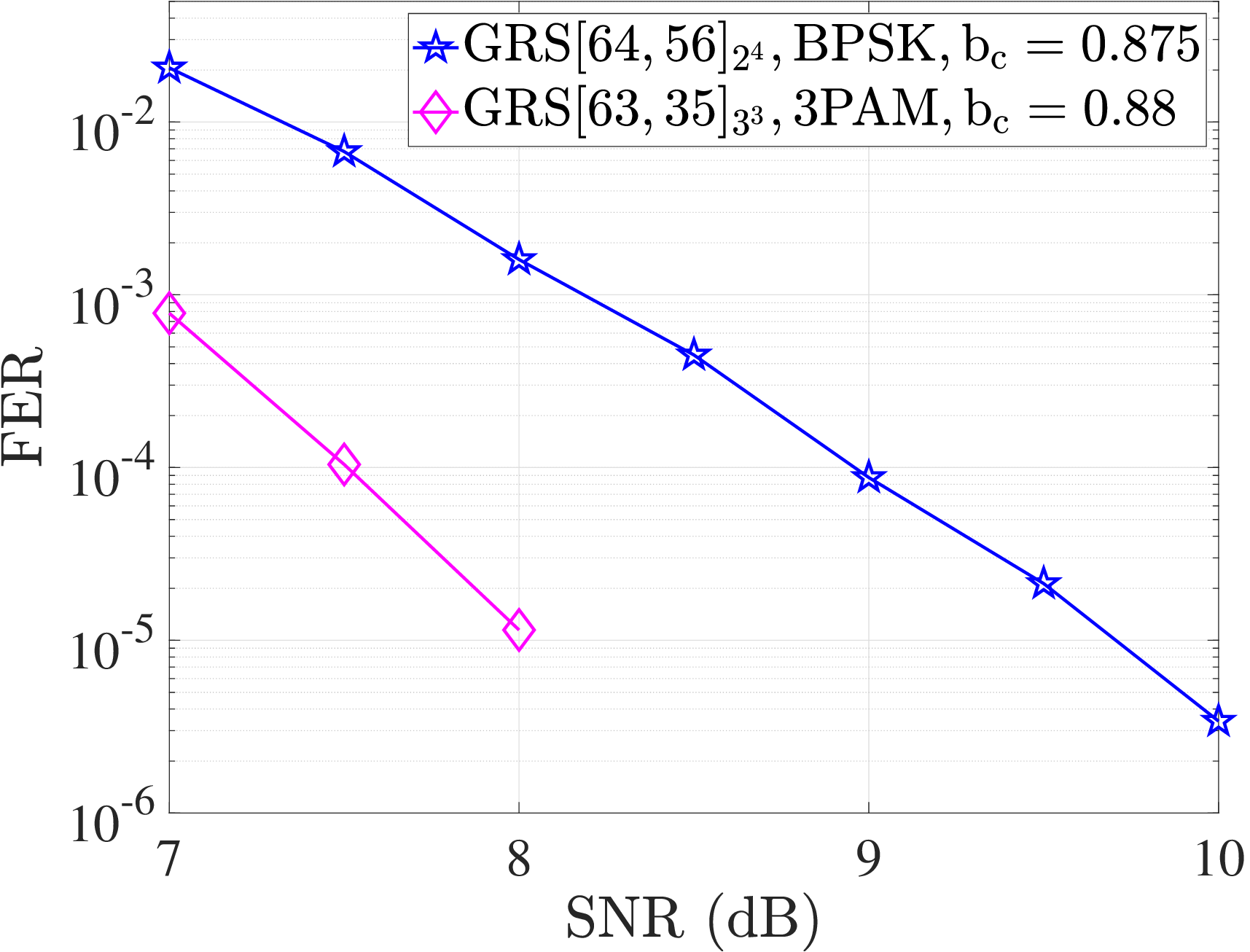}}
	\caption{
	Performance comparison between ternary coding and binary coding with similar bits per channel use $b_c$ and code lengths under LC-OSD. Here, the LC-OSD is with the DAI stopping criterion~\cite{LC_OSDljf2023}, decoding parameters $\delta=4$, and $\ell_{\text{max}}=2^{14}$. (a) Comparison between generalized RS codes $\mathscr{C}[63,32]_{3^3}$ with $b_c = 0.80$ and  $\mathscr{C}[64,51]_{2^4}$ with $b_c = 0.79$. (b) Comparison between generalized RS codes $\mathscr{C}[63,35]_{3^3}$ with $b_c = 0.88$ and  $\mathscr{C}[64,56]_{2^4}$ with $b_c = 0.875$.
	}
	\label{fig_Ternary_source}  
\end{figure}

\begin{figure}[!t]
	\centering
	\includegraphics[width=3.5in]{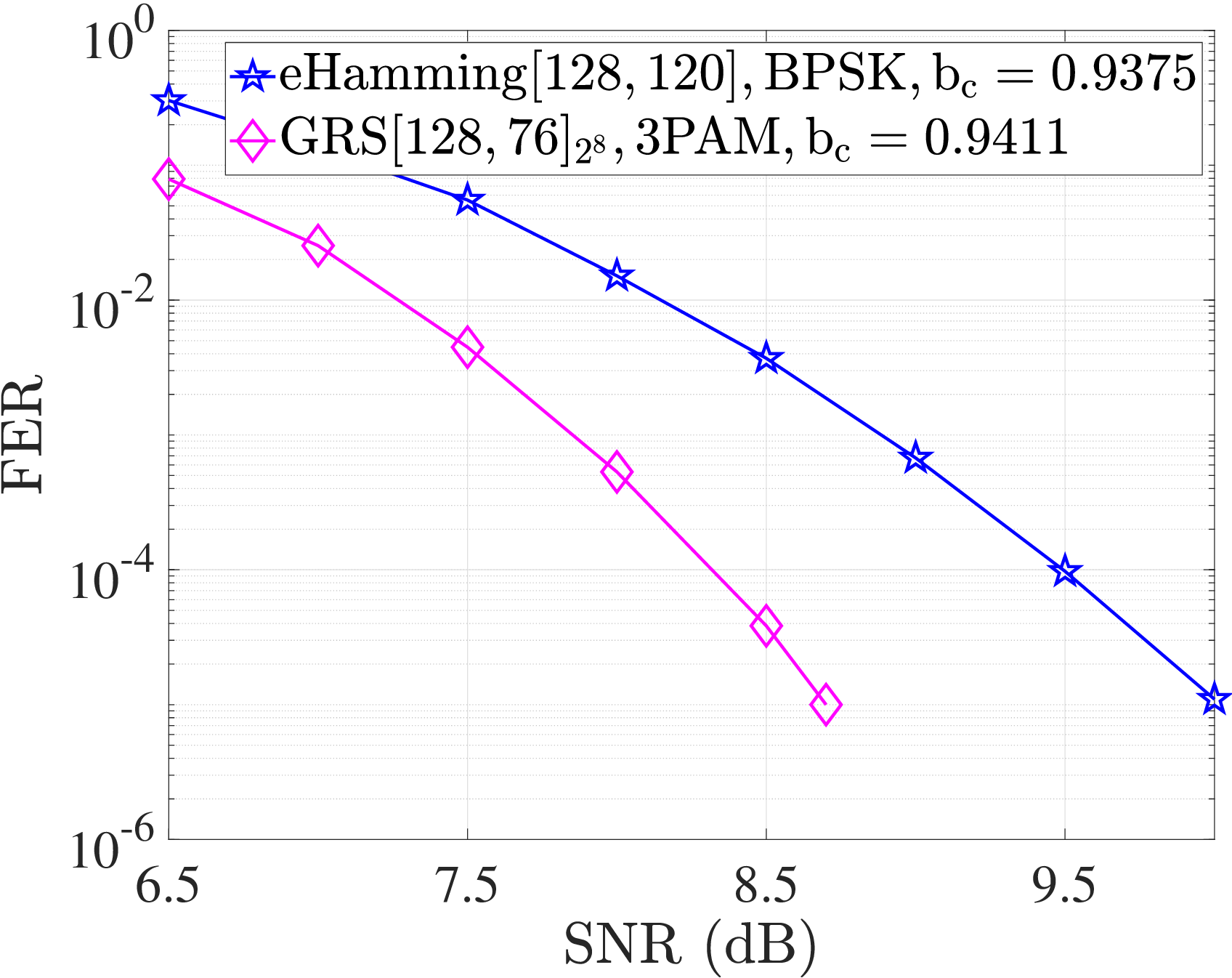}
	\caption{Performance comparison between the ternary generalized RS code $\mathscr{C}_{\text{GRS}}[128,76]_{3^4}$ and the binary extended Hamming code $\mathscr{C}_{\text{Hamming}}[128,120]$  with similar bits per channel use $b_c$ and the same code length $N=128$ under LC-OSD.  Here, the LC-OSD is with the DAI stopping criterion~\cite{LC_OSDljf2023}, decoding parameters $\delta=6$, and $\ell_{\text{max}}=2^{16}$. }
	\label{TernaryHamming} 
\end{figure}

It is worth noting that ternary coding can also applied to binary sources. Consider a binary source vector $\bm v \in \mathbb{F}_2^K$. A simple way to implement ternary coding for binary source is to represent each three bits in $\bm v$ using two trits, delivering a ternary vector $\bm v' \in \mathbb{F}_3^{K'}$ with $K' = 2\lceil K/3\rceil$.  Thus we can encode $\bm v'$ into the ternary image of a generalized RS code for the 3PAM-AWGN channels and employ the ternary LC-OSD algorithm at the receiver. As demonstrated in the following example, ternary coding  tends to outperform binary coding for binary sources, especially when the code rates of the binary coding are high.

\begin{example}
Consider binary sources of length $K \in \{51,56,120\}$, which can be represented as ternary vectors of length $K' \in \{34,38,80\}$. We have presented in Figs.~\ref{fig_Binary_source}-\ref{fig_Binary_sourceHamming} the performance comparison between the ternary and the binary coding for these binary sources under LC-OSD. As we can see,  for comparable code lengths, ternary coding outperforms binary coding for binary sources. Specifically, the generalized RS codes under ternary coding achieve about $0.5$~dB and $1.0$~dB gain at $\rm{FER}=10^{-5}$ for binary sources of length $51$ and $56$, respectively, compared with the corresponding generalized RS codes under binary coding. Compared with the binary extended Hamming code, the ternary generalized RS code can have $1.0$ dB coding gain at $\rm{FER}=10^{-4}$ for binary source of length $120$.
\end{example}

\begin{figure}[!t]
	\centering
	\subfloat[\label{performance_51}]{\includegraphics[width=3.1in]{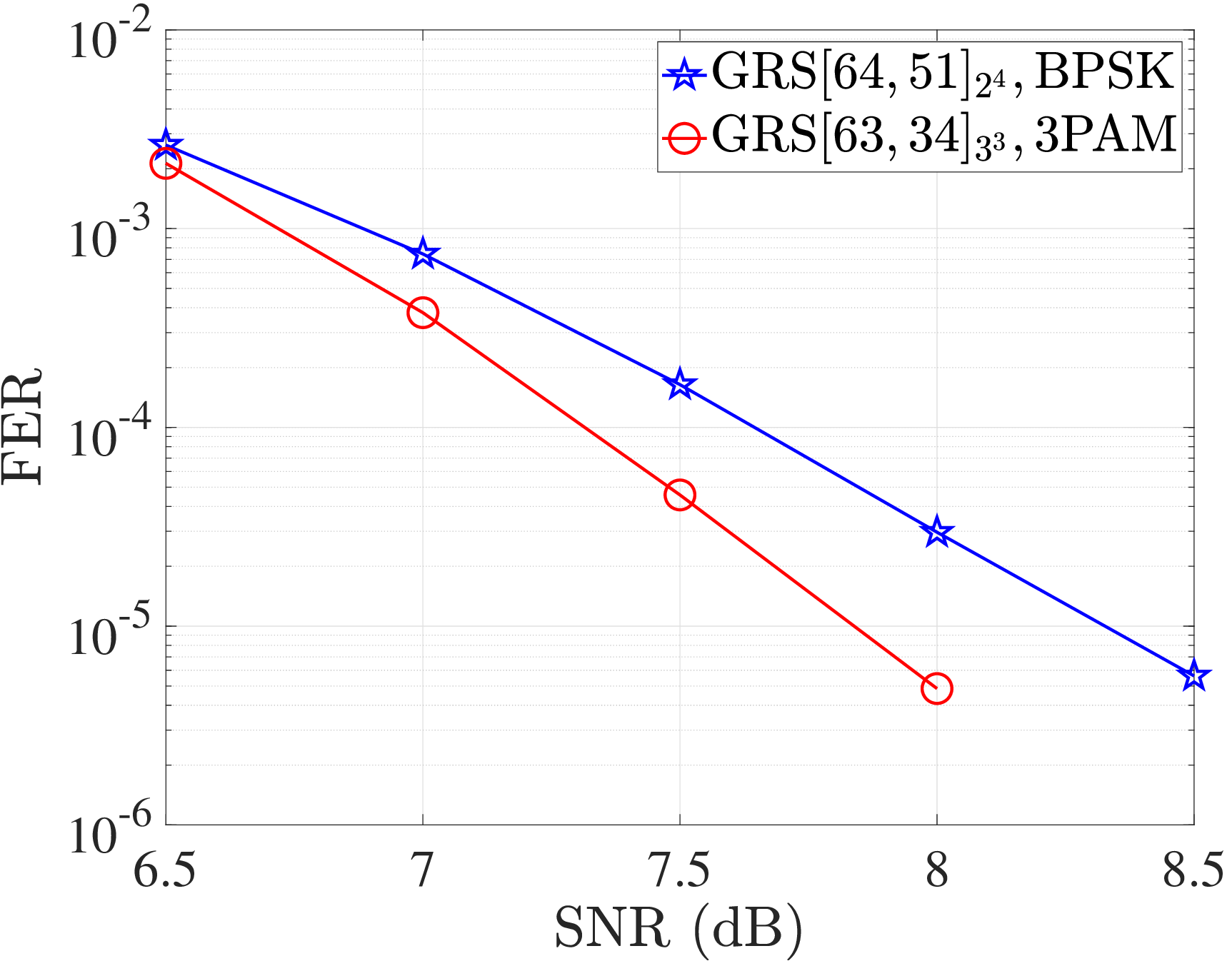}}
	\hfill
	\subfloat[\label{performance_56}]{\includegraphics[width=3.1in]{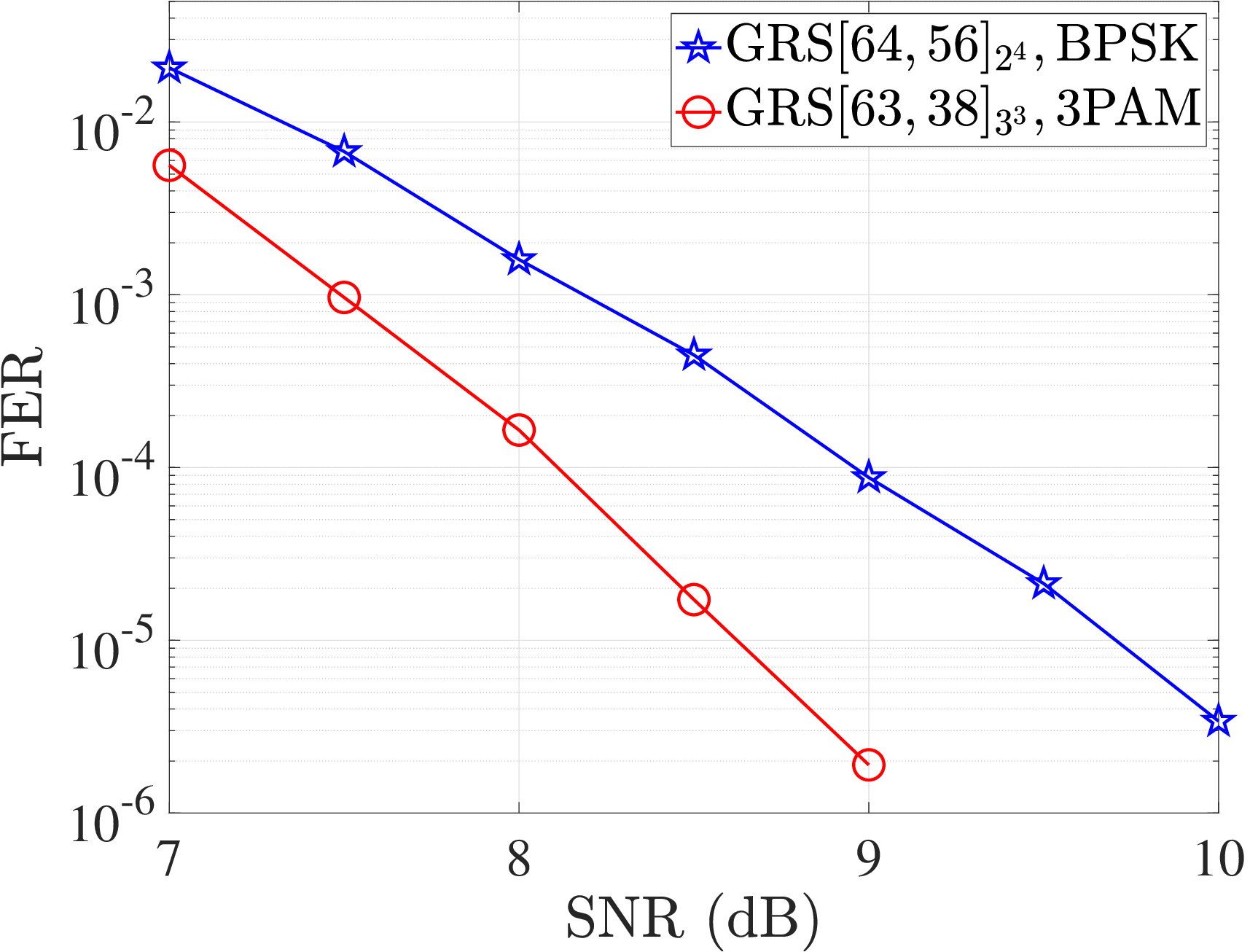}}
	\caption{
		Performance comparison between ternary coding and binary coding for binary sources under LC-OSD. Here, the LC-OSD is with the DAI stopping criterion~\cite{LC_OSDljf2023}, decoding parameters $\delta=4$, and $\ell_{\text{max}}=2^{14}$. (a) Comparison between generalized RS codes $\mathscr{C}[63,34]_{3^3}$ and  $\mathscr{C}[64,51]_{2^4}$ for binary source of length $51$. (b) Comparison between generalized RS codes $\mathscr{C}[63,35]_{3^3}$ and  $\mathscr{C}[64,56]_{2^4}$ for binary source of length $56$.
	}
	\label{fig_Binary_source}  
\end{figure}

\begin{figure}[!t]
	\centering
\includegraphics[width=3.5in]{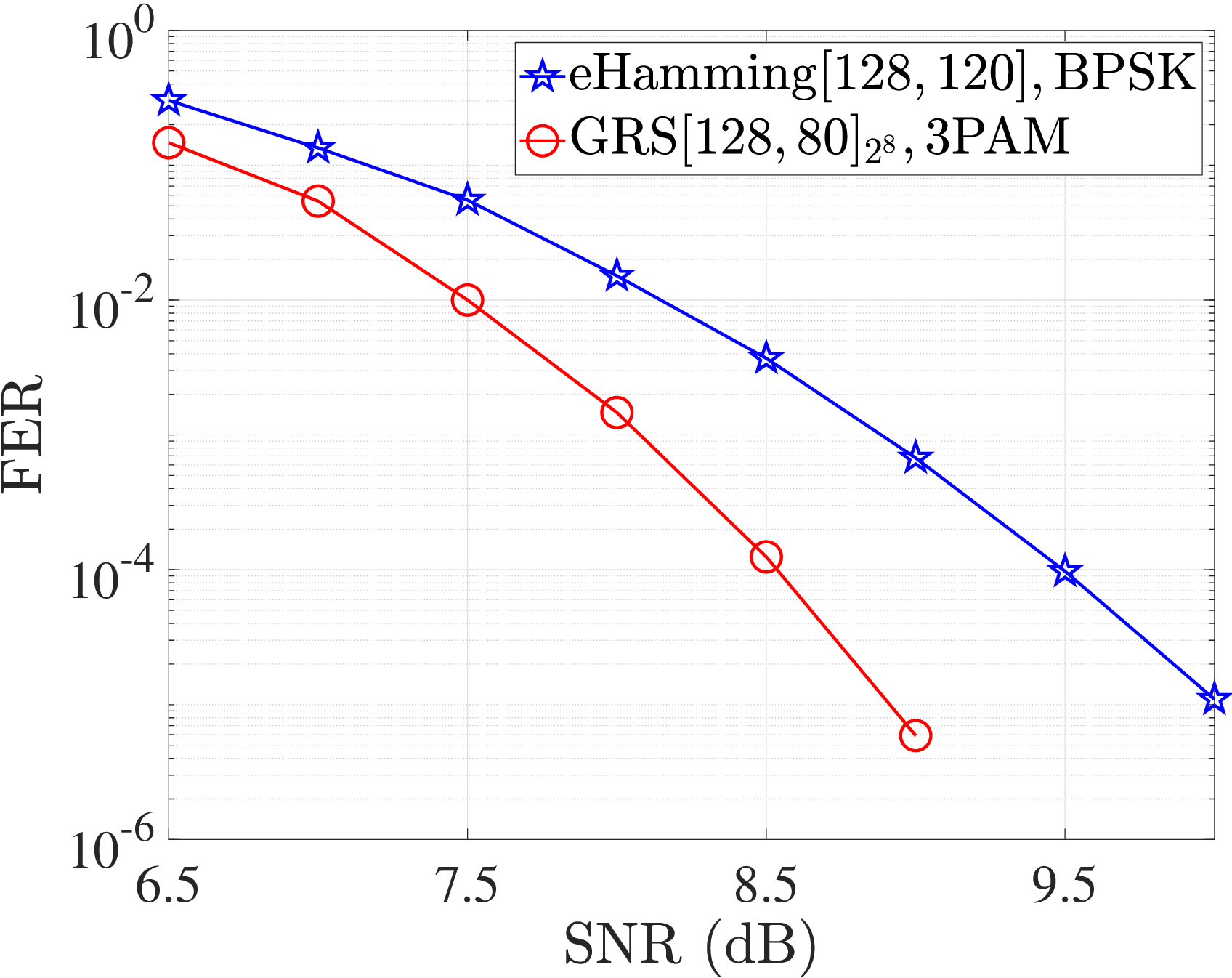}
\caption{Performance comparison between the ternary generalized RS code $\mathscr{C}_{\text{GRS}}[128,80]_{3^4}$ and the binary extended Hamming code $\mathscr{C}_{\text{Hamming}}[128,120]$ for binary source of length $120$ under LC-OSD.  Here, the LC-OSD is with the DAI stopping criterion~\cite{LC_OSDljf2023}, decoding parameters $\delta=6$, and $\ell_{\text{max}}=2^{16}$. }
	\label{fig_Binary_sourceHamming}  
\end{figure}

\section{Conclusions}
We have proved that the generalized RS code is capacity-achieving over $p$-ary symmetric memoryless channels in the infinite-length region. In the finite-length region, we have presented an ML decoding algorithm for generalized RS codes over BECs, where the ML codeword is reconstructed through the change-of-basis on an ordered systematic matrix obtained in parallel by Lagrange interpolation. This implementation accelerates the conventional GE operation and reduces the decoding latency. Additionally, this decoding technique over BECs is applied to the OSD-like algorithms over BPSK-AWGN channels and 3PAM-AWGN channels.

\section*{Acknowledgment}
The first author would like to thank Mr. Jifan Liang from Sun Yat-sen University for his helpful discussions.

\bibliographystyle{IEEEtran}
\bibliography{ref}

\end{document}